\documentclass[twocolumn,secnumarabic,amssymb, aps, prl, superscriptaddress]{revtex4-1}

\usepackage{graphicx}   
\usepackage{import}                         % Graphics package
\usepackage{epstopdf}%%
\usepackage{amsmath} 
\usepackage{bm}
\usepackage{amssymb}
\usepackage{quotes}
\usepackage{color}
\usepackage{transparent}
\usepackage{dsfont}
\usepackage{todonotes}
\usepackage{dcolumn}% Align table columns on decimal point
\usepackage{multirow}
\usepackage{algorithm2e, algorithmic}
\usepackage{natbib}
\usepackage{empheq}
\setcitestyle{super,comma,sort&compress}

\usepackage{cancel} % for the command \cancel and \cancelto
\usepackage{mdframed}
\usepackage{braket}
\usepackage{color}
\usepackage{slashed}
\usepackage[noabbrev,capitalize]{cleveref}
\usepackage{bm}
\usepackage{dsfont}
\usepackage{slashed}
\usepackage{subfigure}
\newtheorem{theorem}{Lemma}
\newtheorem{proof}{Proof}
\usepackage{soul, color} % for the \hl and \st commands
\soulregister\ref{7}  % so that \hl and \st can wrap around \rf
\soulregister\cite{7} % so that \hl and \st can wrap around \cite
\newlength\myindent
\renewcommand*\rmdefault{ptm}
\let\conjugatet\overline

\begin{document}
\rmfamily
%\sffamily

\title{Heuristic Recurrent Algorithms for Photonic Ising Machines}
\author{Charles Roques-Carmes}
\affiliation{Research Laboratory of Electronics, Massachusetts Institute of Technology, 50 Vassar Street, Cambridge MA 02139, USA}
\affiliation{Department of Electrical Engineering and Computer Science, Massachusetts Institute of Technology, 77 Massachusetts Avenue, Cambridge, MA 02139, USA}
\email{Corresponding authors e-mail: chrc@mit.edu, ycshen@mit.edu}

\author{Yichen Shen} 
\affiliation{Department of Physics, Massachusetts Institute of Technology, 77 Massachusetts Avenue, Cambridge, MA 02139, USA}

\author{Cristian Zanoci} 
\affiliation{Department of Physics, Massachusetts Institute of Technology, 77 Massachusetts Avenue, Cambridge, MA 02139, USA}

\author{Mihika Prabhu}  
\affiliation{Research Laboratory of Electronics, Massachusetts Institute of Technology, 50 Vassar Street, Cambridge MA 02139, USA}
\affiliation{Department of Electrical Engineering and Computer Science, Massachusetts Institute of Technology, 77 Massachusetts Avenue, Cambridge, MA 02139, USA}

\author{Fadi Atieh} 
\affiliation{Department of Electrical Engineering and Computer Science, Massachusetts Institute of Technology, 77 Massachusetts Avenue, Cambridge, MA 02139, USA}
\affiliation{Department of Physics, Massachusetts Institute of Technology, 77 Massachusetts Avenue, Cambridge, MA 02139, USA}

\author{Li Jing} 
\affiliation{Department of Physics, Massachusetts Institute of Technology, 77 Massachusetts Avenue, Cambridge, MA 02139, USA}

\author{Tena Dub\v{c}ek} 
\affiliation{Department of Physics, Massachusetts Institute of Technology, 77 Massachusetts Avenue, Cambridge, MA 02139, USA}

\author{Chenkai Mao} 
\affiliation{Department of Electrical Engineering and Computer Science, Massachusetts Institute of Technology, 77 Massachusetts Avenue, Cambridge, MA 02139, USA}
\affiliation{Department of Physics, Massachusetts Institute of Technology, 77 Massachusetts Avenue, Cambridge, MA 02139, USA}

\author{Miles R. Johnson} 
\affiliation{Department of Mathematics, Massachusetts Institute of Technology, 77 Massachusetts Avenue, Cambridge, MA 02139, USA}

\author{Vladimir \v{C}eperi\'{c}} 
\affiliation{Department of Physics, Massachusetts Institute of Technology, 77 Massachusetts Avenue, Cambridge, MA 02139, USA}

\author{John D. Joannopoulos} 
\affiliation{Department of Physics, Massachusetts Institute of Technology, 77 Massachusetts Avenue, Cambridge, MA 02139, USA}
\affiliation{Institute for Soldier Nanotechnologies, 500 Technology Square, Cambridge, MA 02139, USA}

\author{Dirk Englund} 
\affiliation{Research Laboratory of Electronics, Massachusetts Institute of Technology, 50 Vassar Street, Cambridge MA 02139, USA}
\affiliation{Department of Electrical Engineering and Computer Science, Massachusetts Institute of Technology, 77 Massachusetts Avenue, Cambridge, MA 02139, USA}

\author{Marin Solja\v{c}i\'{c}}
\affiliation{Research Laboratory of Electronics, Massachusetts Institute of Technology, 50 Vassar Street, Cambridge MA 02139, USA}
\affiliation{Department of Physics, Massachusetts Institute of Technology, 77 Massachusetts Avenue, Cambridge, MA 02139, USA}

\begin{abstract}
\textbf{The inability of conventional electronic architectures to efficiently solve large combinatorial problems motivates the development of novel computational hardware. There has been much effort recently toward developing novel, application-specific hardware, across many different fields of engineering, such as integrated circuits, memristors, and photonics. However, unleashing the true potential of such novel architectures requires the development of featured algorithms which optimally exploit their fundamental properties. We here present the Photonic Recurrent Ising Sampler (PRIS), a heuristic method tailored for parallel architectures that allows for fast and efficient sampling from distributions of combinatorially hard Ising problems. Since the PRIS relies essentially on vector-to-fixed matrix multiplications, we suggest the implementation of the PRIS in photonic parallel networks, which realize these operations at an unprecedented speed. The PRIS provides sample solutions to the ground state of arbitrary Ising models, by converging in probability to their associated Gibbs distribution. By running the PRIS at various noise levels, we probe the critical behavior of universality classes and their critical exponents. In addition to the attractive features of photonic networks, the PRIS relies on intrinsic dynamic noise and eigenvalue dropout to find ground states more efficiently. Our work suggests speedups in heuristic methods via photonic implementations of the PRIS. We also hint at a broader class of (meta)heuristic algorithms derived from the PRIS, such as combined simulated annealing on the noise and eigenvalue dropout levels. Our algorithm can also be implemented in a competitive manner on fast parallel electronic hardware, such as FPGAs and ASICs.}
\end{abstract}

\maketitle

\section{Introduction}

Heuristic methods -- probabilistic algorithms with stochastic components -- are a cornerstone of both numerical methods in statistical physics \cite{Landau2009} and NP-Hard optimization \cite{Hromkovic2013}. Broad classes of problems in statistical physics, such as growth patterns in clusters \cite{Kardar1986}, percolation \cite{Isichenko1992}, heterogeneity in lipid membranes \cite{Honerkamp-Smith2009}, and complex networks \cite{Albert2002}, can be described by heuristic methods. These methods have proven instrumental for predicting phase transitions and the critical exponents of various universality classes -- families of physical systems exhibiting similar scaling properties near their critical temperature \cite{Landau2009}. These heuristic algorithms have become popular, as they typically outperform exact algorithms at solving real-world problems\cite{Glover2006}. Heuristic methods are usually tailored for conventional electronic hardware; however, a number of optical machines have recently been shown to solve the well-known Ising \cite{Wang2013, McMahon2016} and Traveling Salesman problems \cite{Wu2014, Vazquez2018}. For computationally demanding problems, these methods can benefit from parallelization speedups \cite{Landau2009, Macready1996}, but the determination of an efficient parallelization approach is highly problem-specific \cite{Landau2009}.

Half a century before the contemporary Machine Learning Renaissance \cite{LeCun2015}, the Little \cite{Little1974} and then the Hopfield \cite{Hopfield1982, Hopfield} networks were considered as early architectures of recurrent neural networks (RNN). The latter was suggested as an algorithm to solve combinatorially hard problems, as it was shown to deterministically converge to local minima of arbitrary quadratic Hamiltonians of the form
\begin{equation}
\label{ising}
H^{(K)} = -\frac{1}{2} \sum_{1\leq i,j \leq N} \sigma_i K_{ij}  \sigma_j,
\end{equation}
which is the most general form of an Ising Hamiltonian in the absence of an external magnetic field \cite{Ising1925}. In \cref{ising}, we equivalently denote the set of spins as $\sigma \in \{ -1, 1\}^N$ or $S \in \{ 0, 1\}^N$ (with $\sigma = 2S-1$), and $K$ is a $N \times N$ real symmetric matrix.

In the context of physics, Ising models describe the interaction of many particles in terms of the coupling matrix $K$. These systems are observed in a particular spin configuration $\sigma$ with a probability given by the Gibbs distribution $p(\sigma) \propto \exp (-\beta H^{(K)}(\sigma))$, where $\beta = 1/ (k_B T)$, with $k_B$ the Boltzmann constant and $T$ the temperature. At low temperature, when $\beta \rightarrow \infty$, the Gibbs probability of observing the system in its ground state approaches 1, thus naturally minimizing the quadratic function in Equation (\ref{ising}). As similar optimization problems are often encountered in computer science\cite{Glover2006, Hromkovic2013}, a natural idea is to engineer physical systems with dynamics governed by an equivalent Hamiltonian. Then, by sampling the physical system, one can generate candidate solutions to the optimization problem. This analogy between statistical physics and computer science has nurtured a great variety of concepts in both fields \cite{Mezard2009}, for instance, the analogy between neural networks and spin glasses \cite{Hopfield1982,Amit1985}.

Many complex systems can be formulated using the Ising model\cite{Pelissetto2002} --- such as ferromagnets \cite{Ising1925, Onsager1944}, liquid-vapor transitions \cite{Brilliantov1998}, lipid membranes \cite{Honerkamp-Smith2009}, brain functions \cite{Amit1989}, random photonics\cite{Ghofraniha2015}, and strongly-interacting systems in quantum chromodynamics \cite{Halasz1998}. From the perspective of optimization, finding the spin distribution minimizing $H^{(K)}$ for an arbitrary matrix $K$ belongs to the class of NP-hard problems \cite{Barahona1982}.

Hopfield networks deterministically converge to a \textit{local} minimum, thus making it impossible to scale such networks to \textit{deterministically} find the \textit{global} minimum \cite{Bruck1990} --- thus jeopardizing any electronic \cite{Hopfield} or optical \cite{Farhat1985} implementation of these algorithms. As a result, these early RNN architectures were soon superseded by heuristic (such as Metropolis-Hastings (MH)) and metaheuristic methods (such as simulated annealing (SA) \cite{Kirkpatrick1983}, parallel tempering\cite{Earl2005}, genetic algorithms\cite{Davis1991}, Tabu search\cite{Glover1998} and local-search-based algorithms\cite{Boros2007}), usually tailored for conventional electronic hardware. Even still, heuristic methods struggle to solve large problems, and could benefit from nanophotonic hardware demonstrating parallel, low-energy, and high-speed computations \cite{Shen2017a, Silva2014, Koenderink2015}.

In this Letter, we propose a fast and efficient heuristic method for photonic analog computing platforms, relying essentially on iterative matrix multiplications. Our heuristic approach also takes advantage of optical passivity and dynamic noise to find ground states of arbitrary Ising problems and probe their critical behaviors, yielding accurate predictions of critical exponents of the universality classes of conventional Ising models. Our algorithm presents attractive scaling properties when benchmarked against conventional algorithms, such as MH. Our findings suggest a novel approach to heuristic methods for efficient optimization and sampling by leveraging the potential of matrix-to-vector accelerators, such as parallel photonic networks. \cite{Shen2017a}. Here, we propose a photonic implementation of a passive RNN, which models the arbitrary Ising-type Hamiltonian in \cref{ising}.

\section{Results}

The proposed architecture of our photonic network is shown in Figure \ref{fig:1}. This photonic network can map arbitrary Ising Hamiltonians described by \cref{ising}, with $K_{ii} = 0$ (as diagonal terms only contribute to a global offset of the Hamiltonian, see Supplementary Note 1). In the following, we will refer to the eigenvalue decomposition of $K$ as $K = U D U^\dagger$, where $U$ is a unitary matrix, $U^\dagger$ its transpose conjugate, and $D$ a real-valued diagonal matrix. The spin state at time step $t$, encoded in the phase and amplitude of $N$ parallel photonic signals $S^{(t)} \in \{ 0,1 \}^N$, first goes through a linear symmetric transformation decomposed in its eigenvalue form $2J = U \text{Sq}_\alpha (D) U^\dagger$, where $\text{Sq}_\alpha (D)$ is a diagonal matrix derived from $D$, whose design will be discussed in the next paragraphs. The signal is then fed into nonlinear optoelectronic domain, where it is perturbed by a Gaussian distribution of standard deviation $\phi$ (simulating noise present in the photonic implementation) and is imparted a nonlinear threshold function $\text{Th}_\theta$ ($\text{Th}_\theta (x) = 1$ if $x > \theta$, 0 otherwise). The signal is then recurrently fed back to the linear photonic domain, and the process repeats. The static unit transformation between two time steps $t$ and $t+1$ of this RNN can be summarized as 
\begin{equation}
\begin{aligned}
\label{algo}
X^{(t)} &\sim \mathcal{N}(2JS^{(t)} | \phi), \\
S^{(t+1)} &= \text{Th}_\theta(X^{(t)}) 
\end{aligned}
\end{equation}
where $\mathcal{N}(x | \phi)$ denotes a Gaussian distribution of mean $x$ and standard deviation $\phi$. We call this algorithm, which is tailored for a photonic implementation, the Photonic Recurrent Ising Sampler (PRIS). The detailed choice of algorithm parameters is described in the Supplementary Note 2.

\begin{figure}
\includegraphics[scale=0.4]{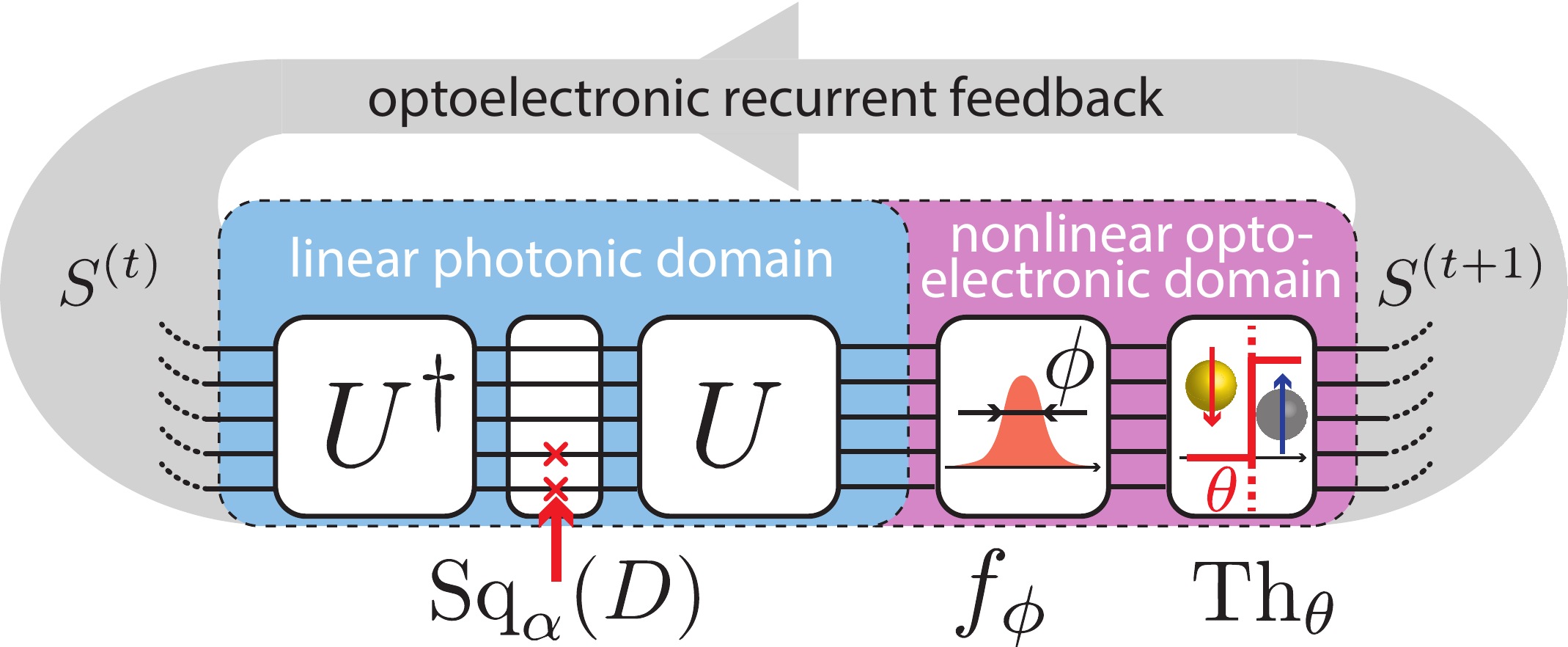}
\caption{\textbf{Operation principle of the PRIS.} A photonic analog signal, encoding the current spin state $S^{(t)}$, goes through transformations in linear photonic and nonlinear optoelectronic domains. The result of this transformation $S^{(t+1)}$ is recurrently fed back to the input of this passive photonic system.}
\label{fig:1}
\end{figure}

This simple recurrent loop can be readily implemented in the photonic domain. For example, the linear photonic interference unit can be realized with MZI networks \cite{Shen2017a, Carolan2015, Reck1994, Clements2016}, diffractive optics \cite{Lin2018, Gruber2000}, ring resonator filter banks \cite{Tait2014, Tait, Vandoorne2014}, and free space lens-SLM-lens systems\cite{Saade2016, Pierangeli2018}; the diagonal matrix multiplication $\text{Sq}_\alpha (D)$ can be implemented with an  electro-optical absorber, a modulator or a single MZI \cite{Cheng2014,Shen2017a, Bao2011a}; the nonlinear optoelectronic unit can be implemented with an optical nonlinearity \cite{Selden1967, Cheng2014,Soljacic2002,Schirmer1997,Bao2011a}, or analog/digital electronics \cite{Horowitz1990, Boser1991, Misra2010, Vrtaric2013}, for instance by converting the optical output to an analog electronic signal, and using this electronic signal to modulate the input \cite{Williamson2019}. The implementation of the PRIS on several photonic architectures and the influence of heterogeneities, phase bit precision, and signal to noise ratio on scaling properties are discussed in the Supplementary Note 5. In the following, we will describe the properties of an ideal PRIS and how design imperfections may affect its performance.

\begin{figure*}
\includegraphics[scale=0.6]{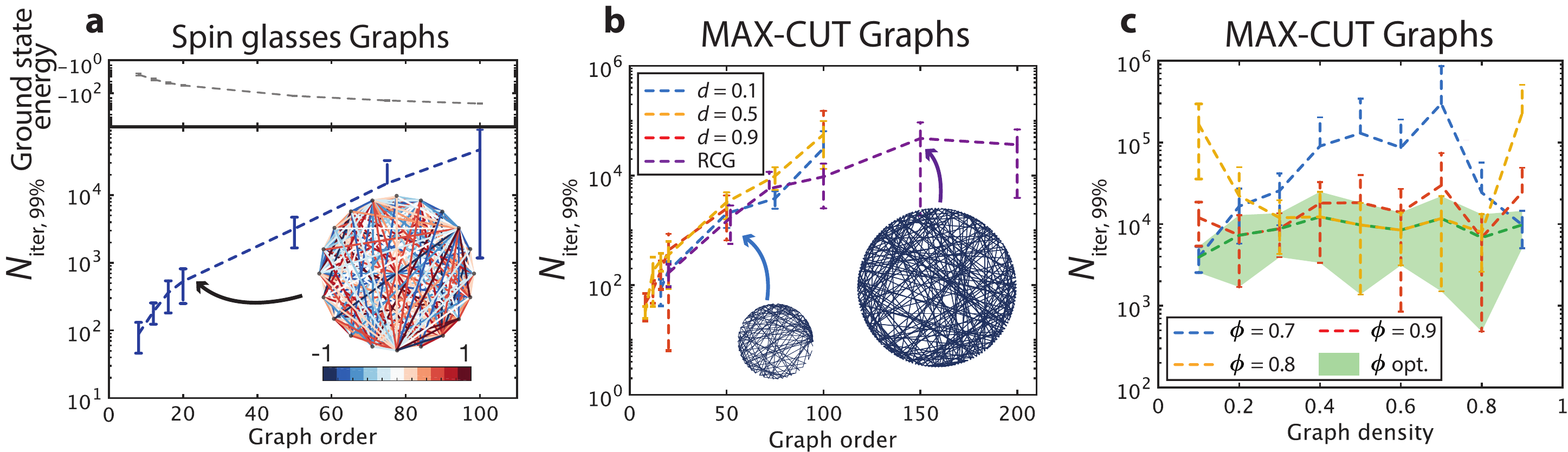}
\caption{\textbf{Scaling performance of the PRIS.}  (\textbf{a}, top) Ground state energy versus graph order of random spin glasses. A sample graph is shown as an inset in (\textbf{a}, bottom): a fully-connected spin glass with uniformly-distributed continuous couplings in $[-1,1]$. $N_{\text{iter}, 99\%}$ versus graph size for spin glasses (\textbf{a}, bottom) and MAX-CUT graphs (\textbf{b}). \textbf{(c)} $N_{\text{iter}, 99\%}$ versus graph density for MAX-CUT graphs and $N = 75$. The graph density is defined as $d = 2 |E| / (N(N-1))$, $|E|$ being the number of undirected edges. RCG denotes Random Cubic Graphs, for which $|E| = 3N/2$. Ground states are determined with the exact solver BiqMac \cite{Rendl2010} (see Methods section). In this analysis, we set $\alpha = 0$, and for each set of density and graph order we ran 10 graph instances 1,000 times. The number of iterations to find the ground state is measured for each run and $N_{\text{iter}, q}$ is defined as the $q$-th quantile of the measured distribution.}
\label{fig:2}
\end{figure*}

\subsection{General theory of the PRIS dynamics}
The long-time dynamics of the PRIS is described by an effective Hamiltonian $H_L$ (see Refs. \cite{Peretto1984, Amit1985} and Supplementary Note 2). This effective Hamiltonian can be computed by performing the following steps. First, calculate the transition probability of a single spin from Equation (\ref{algo}). Then, the transition probability from an initial spin state $S^{(t)}$ to the next step $S^{(t+1)}$ can be written as
\begin{eqnarray}
W^{(0)} \left(S^{(t+1)}|S^{(t)} \right) &=& \frac{e^{-\beta H^0 \left(S^{(t+1)}|S^{(t)} \right)}}{\sum_S e^{-\beta H^0 \left(S|S^{(t)}\right)}}, \label{intermed1} \\
H^0 \left(S|S'\right) &=& - \sum_{1\leq i,j \leq N} J_{ij} \sigma_i \left( S \right) \sigma_j \left( S' \right), \label{intermed2}
\end{eqnarray}
where $S, S'$ denote arbitrary spin configurations. Let us emphasize that, unlike $H^{(K)}(S)$, the transition Hamiltonian $H^{(0)} \left(S|S'\right) $ is a function of two spin distributions $S$ and $S'$. Here, $\beta = 1/(k\phi)$ is analogous to the inverse temperature from statistical mechanics, where $k$ is a constant, only depending on the noise distribution (see Supplementary Table 1). To obtain equations (\ref{intermed1}, \ref{intermed2}), we approximated the single spin transition probability by a rescaled sigmoid function and have enforced the condition $\theta_i = \sum_j J_{ij}$. In the Supplementary Note 2, we investigate the more general case of arbitrary threshold vectors $\theta_i$ and discuss the influence of the noise distribution. 

One can easily verify that this transition probability obeys the triangular condition (or detailed balance condition) if $J$ is symmetric $J_{ij} = J_{ji}$. From there, an effective Hamiltonian $H_L$ can be deduced following the procedure described by Peretto \cite{Peretto1984} for distributions verifying the detailed balance condition. The effective Hamiltonian $H_L$ can be expanded, in the large noise approximation ($\phi \gg 1$, $\beta \ll 1$), into $H_2$:
\color{black}
\begin{eqnarray}
H_L &=& -\frac{1}{\beta} \sum_i \log \cosh \big( \beta \sum_j J_{ij} \sigma_j \big), \label{ham_little}  \\
H_2 &=& -\frac{\beta}{2} \sum_{1 \leq i,j \leq N} \sigma_i [J^2]_{ij}  \sigma_j. \label{ham_large_noise} 
\end{eqnarray}
Examining Equation (\ref{ham_large_noise}), we can deduce a mapping of the PRIS to the general Ising model shown in \cref{ising} since $H_2 = \beta H^{(J^2)}$. %Decomposing the Ising symmetric coupling matrix $K$ into its eigenvalue form $K = U D U^\dagger$, %
	We set the PRIS matrix $J$ to be a modified square-root of the Ising matrix $K$ by imposing the following condition on the PRIS 
\begin{equation}
\text{Sq}_\alpha (D) = 2 \text{Re} ( \sqrt{D + \alpha \Delta} ).
\label{sq_cond}
\end{equation}
We add a diagonal offset term $\alpha \Delta$ to the eigenvalue matrix $D$, in order to parametrize the number of eigenvalues remaining after taking the real part of the square root. Since lower eigenvalues tend to increase the energy, they can be dropped out so that the algorithm spans the eigenspace associated with higher eigenvalues. We chose to parametrize this offset as follows: $\alpha \in \mathbb{R}$ is called the eigenvalue dropout level, a hyperparameter to select the number of eigenvalues remaining from the original coupling matrix $K$, and $\Delta > 0$ is a diagonal offset matrix. For instance, $\Delta$ can be defined as the sum of the off-diagonal term of the Ising coupling matrix $\Delta_{ii} = \sum_{j \neq i} |K_{ij}|$. The addition of $\Delta$ only results in a global offset on the Hamiltonian. The purpose of the $\Delta$ offset is to make the matrix in the square root diagonally dominant, thus  symmetric positive definite, when $\alpha$ is large and positive. Thus, other definitions of the diagonal offset could be proposed. When $\alpha \rightarrow 0$, some lower eigenvalues are dropped out by taking the real part of the square root (see Supplementary Note 3); we show below that this improves the performance of the PRIS. We will specify which definition of $\Delta$ is used in our study when $\alpha \neq 0$. When choosing this definition of $\text{Sq}_\alpha (D)$ and operating the PRIS in the large noise limit, we can implement any general Ising model (\cref{ising}) on the PRIS (Equation (\ref{ham_large_noise})).

It has been noted that by setting $\text{Sq}_\alpha (D) = D $ (i.e. the linear photonic domain matrix amounts to the Ising coupling matrix $2J = K$), the free energy of the system equals the Ising free energy at any finite temperature (up to a factor of 2, thus exhibiting the same ground states) in the particular case of associative memory couplings \cite{Amit1985} with finite number of patterns and in the thermodynamic limit, thus drastically constraining the number of degrees of freedom on the couplings. This regime of operation is a direct modification of the Hopfield network, an energy-based model where the couplings between neurons is equal to the Ising coupling between spins. The essential difference between the PRIS in the configuration $\text{Sq}_\alpha (D) = D$ and a Hopfield network is that the former relies on synchronous spin updates (all spins are updated at every step, in this so-called Little network \cite{Little1974}) while the latter relies on sequential spin updates (a single randomly picked spin is updated at every step). The former is better suited for a photonic implementation with parallel photonic networks.

In this regime of operation, the PRIS can also benefit from computational speed-ups, if implemented on a conventional architecture, for instance if the coupling matrix is sparse. However, as has been pointed out in theory \cite{Amit1985} and by our simulations (see Supplementary Note 4, Figure S7), some additional considerations should be taken into account in order to eliminate non-ergodic behaviors in this system. As the regime of operation described by Equation (\ref{sq_cond}) is general to any coupling, we will use it in the following demonstrations.

\begin{figure*}
\includegraphics[scale=0.6]{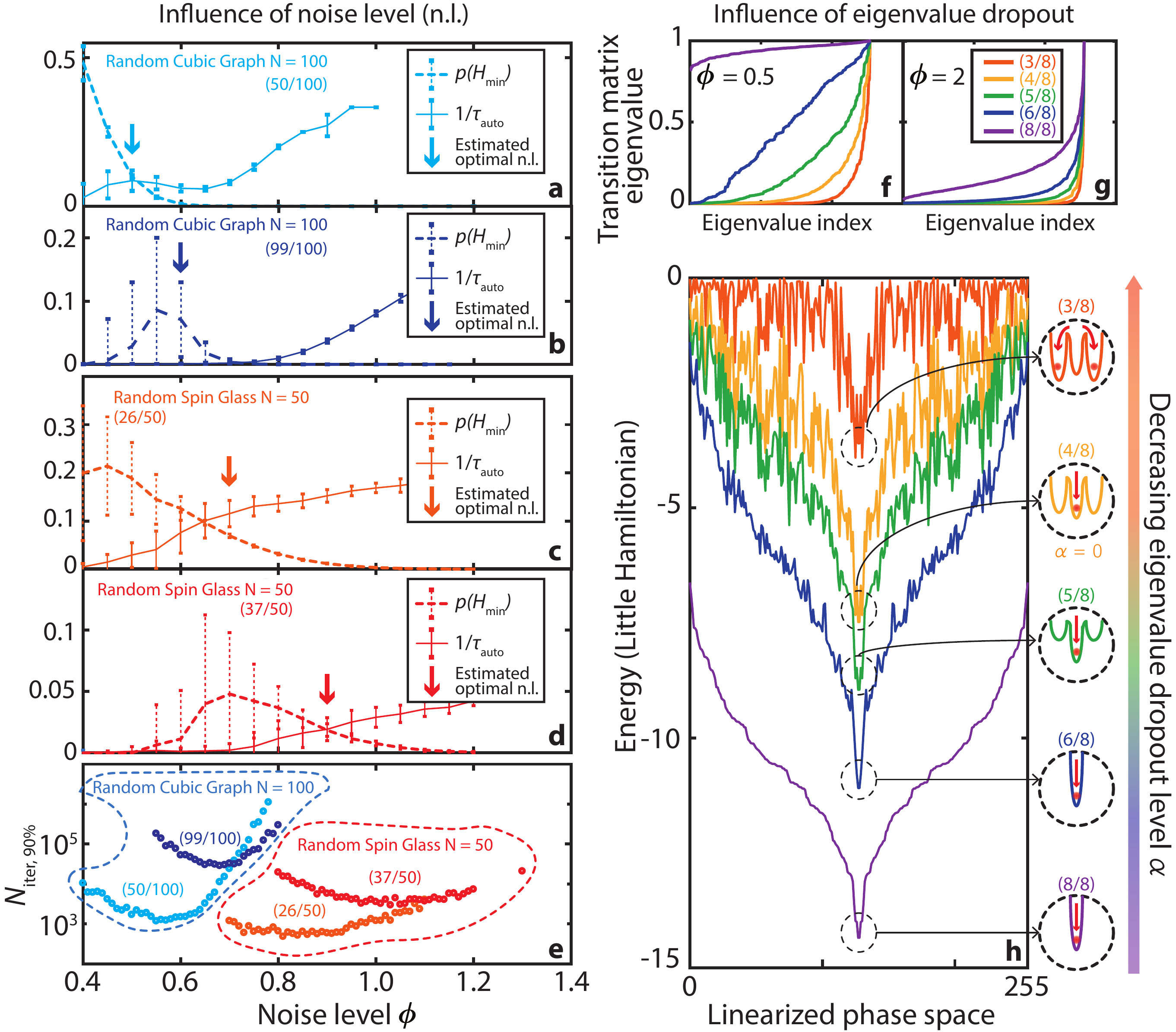}
\caption{\textbf{Influence of noise and eigenvalue dropout levels.} \textbf{(a)-(d):} Probability of finding the ground state, and the inverse of the autocorrelation time as a function of noise level $\phi$ for a sample Random Cubic Graph \cite{McMahon2016} ($N = 100$, (50/100) eigenvalues (\textbf{a}), (99/100) eigenvalues (\textbf{b}) and a sample spin glass ($N = 50$, (37/100) eigenvalues (\textbf{c}), (26/100) eigenvalues (\textbf{d})). The arrows indicate the estimated optimal noise level, from Equation (\ref{Niter}), taking $\tau^E_\text{eq}$ to be constant. For this study we averaged the results of 100 runs of the PRIS with random initial states with error bars representing $\pm \sigma$ from the mean over the 100 runs. We assumed $\Delta_{ii} = | \sum_j K_{ij} |$. \textbf{(e):} $N_{\text{iter}, 90\%}$ versus noise level $\phi$ for these same graphs and eigenvalue dropout levels. \textbf{(f)-(g):} Eigenvalues of the transition matrix of a sample spin glass ($N = 8$) at $\phi = 0.5$ (\textbf{(f)}) and $\phi = 2$ (\textbf{(g)}). \textbf{(h):} The corresponding energy is plotted for various eigenvalue dropout levels $\alpha$, corresponding to less than $N$ eigenvalues kept from the original matrix. The inset is a schematic of the relative position of the global minimum when $\alpha = 1$ (with (8/8) eigenvalues) with respect to nearby local minima when $\alpha < 1$. For this study we assumed $\Delta_{ii} = \sum_j |K_{ij}|$.}
\label{fig:3}
\end{figure*}

\subsection{Finding the ground state of Ising models with the PRIS}
We investigate the performance of the PRIS on finding the ground state of general Ising problems (\cref{ising}) with two types of Ising models: MAX-CUT graphs, which can be mapped to an instance of the unweighted MAX-CUT problem \cite{McMahon2016} and all-to-all spin glasses, whose connections are uniformly distributed in $[-1,1]$ (an example illustration of the latter is shown as an inset in Figure \ref{fig:2}(a)). Both families of models are computationally NP-hard problems \cite{Barahona1982}, thus their computational complexity grows exponentially with the graph order $N$.

The number of steps necessary to find the ground state with $99\%$ probability, $N_{\text{iter}, 99\%}$ is shown in Figure \ref{fig:2}(a-b) for these two types of graphs (see definition in Supplementary Note 4 and in the Methods section). As the PRIS can be implemented with high-speed parallel photonic networks, the on-chip real time of a unit step can be less than a nanosecond \cite{Shen2017a, Lipson2005} (and the initial setup time for a given Ising model is typically of the order of microseconds with thermal phase shifters \cite{Harris2014}). In such architectures, the PRIS would thus find ground states of arbitrary Ising problems with graph orders $N \sim 100$ within less than a millisecond. We also show that the PRIS can be used as a heuristic ground state search algorithm in regimes where exact solvers typically fail ($N \sim 1,000$) and benchmark its performance against MH and conventional metaheuristics (SA) (see Supplementary Note 6). Interestingly, both classical and quantum optical Ising machines have exhibited limitations in their performance related to the graph density \cite{McMahon2016, Hamerly2018}. We observe that the PRIS is roughly insensitive to the graph density, when optimizing the noise level $\phi$ (see Figure \ref{fig:2}(c), shaded green area). A more comprehensive comparison should take into account the static fabrication error in integrated photonic networks \cite{Shen2017a} (see also Supplementary Note 5), even though careful calibration of their control electronics can significantly reduce its impact on the computation \cite{Miller2015a, Burgwal2017}.

\subsection{Influence of the noise and eigenvalue dropout levels}

For a given Ising problem, there remain two degrees of freedom in the execution of the PRIS: the noise and eigenvalue dropout levels. The noise level $\phi$ determines the level of entropy in the Gibbs distribution probed by the PRIS $p(E) \propto \exp (-\beta (E - \phi S(E)))$, where $S(E)$ is the Boltzmann entropy associated with the energy level $E$. On the one hand, increasing $\phi$ will result in an exponential decay of the probability of finding the ground state $p(H_\text{min}, \phi)$. On the other hand, too small a noise level will not satisfy the large noise approximation $H_L \sim H_2$ and result in large autocorrelation times (as the spin state could get stuck in a local minimum of the Hamiltonian). Figure \ref{fig:3}(e) demonstrates the existence of an optimal noise level $\phi$, minimizing the number of iterations required to find the ground state of a given Ising problem, for various graph sizes, densities, and eigenvalue dropout levels. This optimal noise value can be approximated upon evaluation of the probability of finding the ground state $p(H_\text{min}, \phi)$ and the energy autocorrelation time $\tau^E_\text{auto}$, as the minimum of the following heuristic 
\begin{equation}
N_{\text{iter}, q} \sim \tau^E_\text{eq}(\phi) + \tau^E_\text{auto}(\phi) \frac{\log (1-q)}{\log (1-p(H_\text{min}, \phi))},
\label{Niter}
\end{equation}
which approximates the number of iterations required to find the ground state with probability $q$ (see Figure \ref{fig:3}(a-e)). In this expression, $\tau^E_\text{eq}(\phi)$ is the energy equilibrium (or burn-in) time. As can be seen in Figure \ref{fig:3}(e), decreasing $\alpha$ (and thus dropping more eigenvalues, with the lowest eigenvalues being dropped out first) will result in a smaller optimal noise level $\phi$. Comparing the energy landscape for various eigenvalue dropout levels (Figure \ref{fig:3}(h)) confirms this statement: as $\alpha$ is reduced, the energy landscape is perturbed. However, for the random spin glass studied in Figure \ref{fig:3}(f-g), the ground state remains the same down to $\alpha = 0$. This hints at a general observation: as  lower eigenvalues tend to increase the energy, the Ising ground state will in general be contained in the span of eigenvectors associated with higher eigenvalues (see discussion in the Supplementary Note 3). Nonetheless, the global picture is more complex, as the solution of this optimization problem should also enforce the constraint $\sigma \in \{ -1 , 1\}^N$. We observe in our simulations that $\alpha = 0$ yields a higher ground state probability and lower autocorrelation times than $\alpha > 0$ for all the Ising problems we used in our benchmark. In some sparse models, the optimal value can even be $\alpha < 0$ (see Figure S3 in the Supplementary Note 4). The eigenvalue dropout is thus a parameter that constrains the dimensionality of the ground state search.

The influence of eigenvalue dropout can also be seen from the perspective of the transition matrix. Figure \ref{fig:3}(f-g) shows the eigenvalue distribution of the transition matrix for various noise and eigenvalue dropout levels. As the PRIS matrix eigenvalues are dropped out, the transition matrix eigenvalues become more nonuniform, as in the case of large noise (Figure \ref{fig:3}(g)). Overall, the eigenvalue dropout can be understood as a means of pushing the PRIS to operate in the large noise approximation, without perturbing the Hamiltonian in such a way that would prevent it from finding the ground state. The improved performance of the PRIS with $\alpha \sim 0$ hints at the following interpretation: the perturbation of the energy landscape (which affects $p(H_\text{min})$) is counterbalanced by the reduction of the energy autocorrelation time induced by the eigenvalue dropout. The existence of these two degrees of freedom suggests a realm of algorithmic techniques to optimize the PRIS operation. One could suggest, for instance, setting $\alpha \approx 0$, and then performing an inverse simulated annealing of the eigenvalue dropout level to increase the dimensionality of the ground state search. This class of algorithms could rely on the development of high-speed, low-loss integrated modulators \cite{Almeida2004, Lipson2005, Phare2015, Haffner2018}.

\section{Discussion}
\subsection{Detecting and characterizing critical behaviors with the PRIS}

The existence of an effective Hamiltonian describing the PRIS dynamics (Equation (\ref{ham_large_noise})) further suggests the ability to generate samples of the associated Gibbs distribution at \textit{any finite temperature}. This is particularly interesting considering the various ways in which noise can be added in integrated photonic circuits by tuning the operating temperature, laser power, photodiode regimes of operation, etc. \cite{Horowitz1990, Hamerly2018b}. This alludes to the possibility of detecting phase transitions and characterizing critical exponents of universality classes, leveraging the high speed at which photonic systems can generate uncorrelated heuristic samples of the Gibbs distribution associated with Equations (\ref{ham_little},\ref{ham_large_noise}). In this part, we operate the PRIS in the regime where the linear photonic matrix is equal to the Ising coupling matrix ($\text{Sq}_\alpha (D) = D$) \cite{Amit1985}. This allows us to speedup the computation on a CPU by leveraging symmetry and sparsity of the coupling matrix $K$. We show that the regime of operation described by Equation (\ref{sq_cond}) also probes the expected phase transition (see Supplementary Note 4).

A standard way of locating the critical temperature of a system is through the use of the Binder cumulant \cite{Landau2009} $U_4(L) = 1 - \langle m^4 \rangle / (3 \langle m^2 \rangle^2)$, where $m = \sum_{i=1}^N \sigma_i / N$ is the magnetization and $\langle . \rangle$ denotes the ensemble average. As shown in Figure \ref{fig:4}(a), the Binder cumulants intersect for various graph sizes $L^2 = N$ at the critical temperature of $ T_C = 2.241$ (compared to the theoretical value of 2.269 for the two-dimensional Ferromagnetic Ising model, i.e. within 1.3\%). The heuristic samples generated by the PRIS can be used to compute physical observables of the modeled system, which exhibit the emblematic order-disorder phase transition of the two-dimensional Ising model \cite{Onsager1944, Landau2009} (Figure \ref{fig:4}(b)). In addition, critical parameters describing the scaling of the magnetization and susceptibility at the critical temperature can be extracted from the PRIS to within $10\%$ of the theoretical value (see Supplementary Note 4). 

\begin{figure}
\includegraphics[scale=0.38]{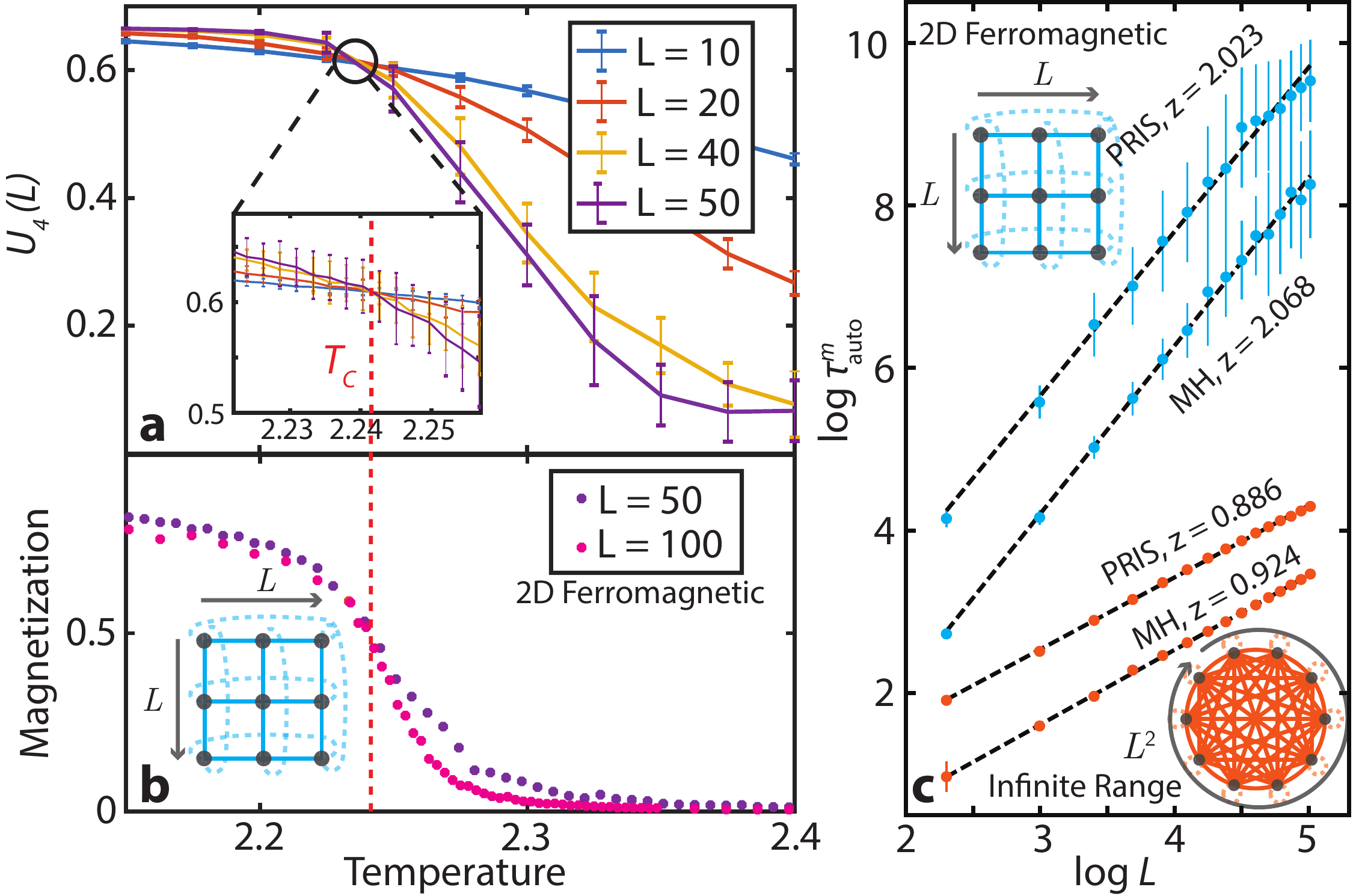}
\caption{\textbf{Detecting and characterizing critical behaviors. a:} Binder cumulants $U_4 (L)$ for various graph sizes L on the 2D Ferromagnetic Ising model. Their intersection determines the critical temperature of the model $T_C$ (denoted by a dotted line). \textbf{b:} Magnetization estimated from the PRIS for various $L$. \textbf{c:} Scaling of the PRIS magnetization autocorrelation time for various Ising models, benchmarked versus the Metropolis-Hastings algorithm (MH). The complexity of a single time step scales like $N^2 = L^4$ for MH on a CPU and like $N = L^2$ for the PRIS on a photonic platform. For readability, error bars in (b) are not shown (see Supplementary Note 4).}
\label{fig:4}
\end{figure}

In Figure \ref{fig:4}(c), we benchmark the performance of the PRIS against the well-known Metropolis-Hastings (MH) algorithm \cite{Landau2009, Metropolis1953, Hastings1970}. In the context of heuristic methods, one should compare the autocorrelation time of a given observable. The scaling of the magnetization autocorrelation time $\tau^m_\text{auto} = \mathcal{O}(L^z) = \mathcal{O}(N^{z/2})$ at the critical temperature is shown in Figure \ref{fig:4}(c) for two analytically-solvable models: the two-dimensional ferromagnetic and the infinite-range Ising models. Both algorithms yield autocorrelation time critical exponents close to the theoretical value ($z \sim 2.1$) \cite{Landau2009} for the two-dimensional Ising model. However, the PRIS seems to perform better on denser models such as the infinite-range Ising model, where it yields a smaller autocorrelation time critical exponent. More significantly, the advantage of the PRIS resides in its possible implementation with any matrix-to-vector accelerator, such as parallel photonic networks, so that the computational (time) complexity of a single step is $\mathcal{O}(N)$ \cite{Shen2017a, Clements2016, Reck1994}. Thus, the computational complexity of generating an uncorrelated sample scales like $\mathcal{O}(N^{1+z_\text{PRIS}/2})$ for the PRIS on a parallel architecture, while it scales like $\mathcal{O}(N^{2+z_\text{MH}/2})$ for a sequential implementation of MH, on a CPU for instance. Implementing the PRIS on a photonic parallel architecture also ensures that the prefactor in this order of magnitude estimate is small (and only limited by the clock rate of a single recurrent step of this high-speed network). Thus, as long as $z_\text{PRIS} < z_\text{MH} + 2$, the PRIS exhibits a clear advantage over MH implemented on a sequential architecture.

\subsection{Conclusion}
To conclude, we presented in this Letter the PRIS, a photonic-based heuristic algorithm able to probe arbitrary Ising Gibbs distributions at various temperature levels. At low temperatures, the PRIS can find ground states of arbitrary Ising models with high probability. Our approach essentially relies on the use of matrix-to-vector product accelerators, such as photonic networks \cite{Shen2017a, Hamerly2018b}, free-space optical processors \cite{Farhat1985}, FPGAs \cite{Dean2018}, and ASICs \cite{Dou2005} (see comparison of time estimates in the Supplementary Note 5). We also perform a proof-of-concept experiment on a Xilinx Zynq UltraScale+ multiprocessor system-on-chip (MPSoC) ZCU104, an electronic board containing a parallel programmable logic unit (FPGA – Field Programmable Gate Arrays). We run the PRIS on large random spin glasses $N=100$ and achieve algorithm time steps of 63 ns. This brings us closer to photonic clocks $\lesssim 1$ ns, thus demonstrating that (1) the PRIS can leverage parallel architectures of various natures, electronics and photonics; (2) the potential of hybrid parallel opto-electronic implementations. Details of the FPGA implementation and numerical experiments are given in Supplementary Note 7.

Moreover, our system requires some amount of noise to perform better, which is an unusual behavior only observed in very few physical systems. For instance, neuroscientists have conjectured that this could be a feature of the brain and spiking neural networks \cite{Knill2004, Maass2014}. The PRIS also performs a static transformation (and the state evolves to find the ground state). This kind of computation can rely on a fundamental property of photonics --- passivity --- and thus reach even higher efficiencies. Non-volatile phase-change materials integrated in silicon photonic networks could be leveraged to implement the PRIS with minimal energy costs \cite{Wang2016}.

We also suggested a broader family of photonic metaheuristic algorithms which could achieve even better performance on larger graphs (see Supplementary Note 6). For instance, one could simulate annealing with photonics by reducing the system noise level (this could be achieved by leveraging quantum photodetection noise \cite{Hamerly2018b}, see discussion in Supplementary Notes 5 and 6). We believe that this class of algorithms that can be implemented on photonic networks is broader than the metaheuristics derived from MH, since one could also simulate annealing on the eigenvalue dropout level $\alpha$.

The ability of the PRIS to detect phase transitions and probe critical exponents is particularly promising for the study of universality classes, as numerical simulations suffer from critical slowing down: the autocorrelation time grows exponentially at the critical point, thus making most samples too correlated to yield accurate estimates of physical observables. Our study suggests that this fundamental issue could be bypassed with the PRIS, which can generate a very large number of samples per unit time -- only limited by the bandwidth of active silicon photonics components.

The experimental realization of the PRIS on a photonic platform would require additional work compared to the demonstration of deep learning with nanophotonic circuits \cite{Shen2017a}. The noise level can be dynamically induced by several well-known sources of noise in photonic and electronic systems \cite{Horowitz1990}. However, attaining a low enough noise due to heterogeneities in a static architecture, and characterizing the noise level are two experimental challenges. Moreover, the PRIS requires an additional homodyne detection unit, in order to detect both the amplitude and the phase of the output signal from the linear photonic domain. Nonetheless, these experimental challenges do not impact the promising scaling properties of the PRIS, since various photonic architectures have recently been proposed\cite{Saade2016, Shen2017a, Tait2017, Lin2018, Hamerly2018b}, giving a new momentum to photonic computing.

\section{Methods}

\subsection{Numerical simulations to evaluate performance on finding ground states of Ising models}

To evaluate the performance of the algorithm on several Ising problems, we simulate the execution of an ideal photonic system, performing computations without static error. The noise is artificially added after the matrix multiplication unit and follows a Gaussian distribution, as discussed above. This results in an algorithm similar to the one des¡cribed in the section II of this work.

In the main text, we present the scaling performance of the PRIS as a function of the graph order. For each graph order and density, we generate 10 random samples with these properties. We then optimize the noise level (minimizing $N_{\text{iter}, 99\%}$) on a random sample graph and generate a total of 10 samples for each pair of graph order/density. The optimal value of $\phi$ is shown in Figure S2 in Supplementary Note 4.

For each randomly generated graph, we first compute its ground state with the online platform BiqMac \cite{Rendl2010}. We then make 100 measurements of the number of steps required (with a random initial state) to get to this ground state. From these 1000 runs, we define the estimate of finding the ground state of the problem with $q$ percent probability $N_{\text{iter}, q\%}$ as the $q$-th quantile.

Also in the main text, we study the influence of eigenvalue dropout and of the noise level on the PRIS performance. We show that the optimal level of eigenvalue dropout is usually $\alpha <1$, and around $\alpha = 0$. In some cases, it can even be $\alpha < 0$ as we show in Figure S3 in Supplementary Note 4 where the optimal $(\alpha, \phi) = (-0.15, 0.55)$ for a random cubic graph with $N = 52$. In addition to Figure 3(f-h) from the main text showing the influence of eigenvalue dropout on a random spin glass, the influence of dropout on a random cubic graph is shown in Figure S4 in Supplementary Note 4. Similar observations can be made, but random cubic graphs, which show highly degenerated hamiltonian landscapes, are more robust to eigenvalue dropout. Even with $\alpha = -0.8$, in the case shown in Figure S4 in Supplementary Note 4 the ground state remains unaffected.

\subsection{Others}
Further details on generalization of the theory of the PRIS dynamics, construction of the weight matrix $J$, numerical simulations, scaling performance of the PRIS, and comparison of the PRIS to other (meta)heuristics algorithms can be found in the Supplementary Information.

\color{black}
\section{Acknowledgements} The authors would like to acknowledge Aram Harrow, Mehran Kardar, Ido Kaminer, Miriam Farber, Theodor Misiakiewicz, Manan Raval, Nicholas Rivera, Nicolas Romeo, Jamison Sloan, Can Knaut, Joe Steinmeyer, and Gim P. Hom for helpful discussions. The authors would also like to thank Angelika Wiegele (Alpen-Adria-Universität Klagenfurt) for providing solutions of the Ising models considered in this work with $N\geq 50$ (computed with BiqMac \cite{Rendl2010}). This work was supported in part by the Semiconductor Research Corporation (SRC) under SRC contract \#2016-EP-2693-B (Energy Efficient Computing with Chip-Based Photonics – MIT). This work was supported in part by the National Science Foundation (NSF) with NSF Award \#CCF-1640012 (E2DCA: Type I: Collaborative Research: Energy Efficient Computing with Chip-Based Photonics). This material is based upon work supported in part by the U. S. Army Research Laboratory and the U.S. Army Research Office through the Institute for Soldier Nanotechnologies, under contract number W911NF-18-2-0048. C. Z. was financially supported by the Whiteman Fellowship. M. P. was financially supported by NSF Graduate Research Fellowship grant number 1122374.

\section{Author contributions}
C. R.-.C., Y. S. and M.S. conceived the project. C. R.-C. and Y. S. developed the analytical models and numerical calculations, with contributions from C. Z., M. P., L. J. and T. D.; C. R.-C. and C. Z. performed the benchmarking of the PRIS on analytically solvable Ising models and large spin glasses. C. R.-C. and F.A. developed the analytics for various noise distributions. C. M., M. R. J., and C. R.-C. implemented the PRIS on FPGA. Y. S., J. D. J., D. E. and M.S. supervised the project. C. R.-C. wrote the manuscript with input from all authors.

\section{Additional information}
Correspondence and requests for materials should be addressed to C. R.-C. 

\section{Competing financial interests}
The authors declare the following patent application: U.S. Patent Application No.: 16/032,737.

\bibliographystyle{ieeetr}
\bibliography{ising, resetnumbers=true}
%\nocite{*}

\onecolumngrid
\renewcommand*\rmdefault{ptm}
\setcounter{figure}{0}
\let\conjugatet\overline

\newpage 
\section{Heuristic Recurrent Algorithms for Photonic Ising Machines \\ SUPPLEMENTARY INFORMATION}

\newpage

\maketitle

\setcounter{tocdepth}{1}
\tableofcontents

\newpage 

In this work, we first provide a comprehensive theoretical analysis of the dynamics of the Photonic Recurrent Ising Sampler (PRIS), with an emphasis on how the weight matrix in the algorithm should be chosen. We then detail the parameters of our numerical simulations carried to estimate the performance of the PRIS on finding the ground state of arbitrary Ising problems and to sample critical behaviors.

\section{Supplementary Note 1: Theoretical preliminaries}

\subsection{Definition of an Ising problem}
The general definition of an Ising problem is the following: given a matrix $(K_{ij})_{(i,j) \in I^2}$, and a vector $\{ b_i \}_{i \in I}$ where $I$ is a finite set with cardinality $|I| = N$, we want to find the spin distribution $\{ \sigma_i\}_{i \in I} \in \{-1, 1\}^{N}$ minimizing the following Hamiltonian function:
\begin{eqnarray}
H_{K,b} (\{ \sigma_i\})  &=& -\sum_{1 \leq i<j \leq N}  K_{ij} \sigma_i \sigma_j - \sum_i b_i \sigma_i \\
&=& -\frac{1}{2} \sum_{1 \leq i,j \leq N}  K_{ij} \sigma_i \sigma_j - \sum_i b_i \sigma_i.
\label{ising_general}
\end{eqnarray}

In statistical mechanics, this model represents the energy of an interacting set of spins. The coupling term $K_{ij}$ represents the coupling between spins $i$ and $j$. In the general definition of the Ising model, the coupling matrix $K$ is assumed to be symmetric. We could also assume that it has a null diagonal, as it would only add a constant offset to the Hamiltonian. We are not making this extra assumption, as it will prove later to be useful. An external magnetic field $b_i$ can be applied, which breaks the symmetry of the problem. This general class of Ising problems is NP-hard \cite{Karp1972} and many subclasses of this problem exhibit a similar complexity. In this work, we will characterize our optical algorithms on two subclasses of problems:
\begin{enumerate}
\item[$\triangleright$] General antiferromagnets, for which the coupling $K$ can only take two discrete values $K_{ij} \in \{0, -1\}$, and $b_i = 0$ for all $i$. As every problem in this subclass is equivalent to an unweighted MAX-CUT problem, and that MAX-CUT is known to be NP-hard \cite{Karp1972}, this subclass is already NP-hard. 
\item[$\triangleright$] Spin glasses, for which the coupling $K$ can take on continuous values in the range $[-1,1]$, and $b_i = 0$. 
\end{enumerate}
However, our approach can be easily extended to any coupling matrix $K$. The possible extension to non-zero external magnetic fields will also be discussed. A complete study of the hardness of the various subclasses of Ising problems can be found in Ref.\cite{Barahona1982}.

\subsection{Mapping to the MAX-CUT problem}

The MAX-CUT problem of a weighted undirected graph can be phrased in terms of its adjacency matrix $A$:
\begin{equation}
\text{Find} \ \ \text{argmax}_\sigma \frac{1}{4} \sum_{ij} A_{ij} (1 - \sigma_i \sigma_j)  := \text{argmax}_\sigma C_A(\sigma),
\end{equation}
where $\sigma \in \{-1, 1\}^N$ is a spin vector. The value of this MAX-CUT problem is called $C_\text{max}$. More intuitively, this problem can be interpreted as finding a subset of vertices of the graph, such that the number of edges connecting this subset to its complementary is maximized. The vertices of this subset (resp. of its complementary set) will have spin up (resp. spin down).

We can map the general Ising problem (spin glass) (Equation (\ref{ising_general})) to any weighted MAX-CUT problem. This mapping is useful because the publicly-available solver we use to find the Ising ground state works in terms of the MAX-CUT problem \cite{Rendl2010}. The energy of the Ising ground state and the MAX-CUT are related as follows:
\begin{equation}
H_{\text{min}} = - \frac{1}{2} \sum_{ij} K_{ij} - 2C_{\text{max}},
\end{equation}
where $K = - A$. From this linear relation, we also deduce that the solution of the weighted MAX-CUT problem and of the general Ising model (spin glass) are the same:
\begin{equation}
\text{argmax}_\sigma C_A(\sigma) = \text{argmin}_\sigma H^{(K)}(\sigma).
\end{equation}

\newpage

\section{Supplementary Note 2: General theory of Photonic Recurrent Ising Sampler}

\subsection{Optical Parallel Recurrent MCMC algorithm}
The Little and Hopfield networks \cite{Little1974, Hopfield1982} are two early forms of recurrent neural networks, originally designed to understand and model associative memory processes. In their general acceptance, the Little network can be understood as a synchronous version of the Hopfield network. Using the latter to solve the NP-hard Ising problem was previously investigated \cite{Hopfield, Looi1992}. The behavior of both networks was later studied from a statistical mechanical perspective by P. Peretto \cite{Peretto1984} and Daniel J. Amit and colleagues \citep{Amit1985}. We here generalize their study to the case of noisy, synchronous Hopfield network, or noisy Little network \cite{Peretto1984} and investigate the possibility of sampling hard Ising problems with a parallel, recurrent optical system. In the following, we equivalently use the term "neuron" and "spin" in the framework of spin glass models of neural networks.

Our machine is applying the following algorithm:
\begin{enumerate}
\item[$\triangleright$] Initialize assembly of neurons $\{ \sigma_i \}$ with random values. The signal is coded into optical domain with the reduced spins $S_i = (\sigma_i+1)/2 \in \{0,1 \}$.
\item[$\triangleright$] At each step of the algorithm, each neuron is applied a potential $v_i$ (random variable) whose expected value is given by a matrix multiplication with $C = 2J$:
\begin{equation}
\bar{v}_i = \sum_j C_{ij} S_j = \sum_{j=1}^N J_{ij} \sigma_j + \sum_{j=1}^N  C_{ij}/2
\end{equation}
and whose probability density is given by the function $f_\phi$ with mean $\bar{v_i}$ and standard deviation $\phi$:
\begin{eqnarray}
\int f_\phi(x) \  x  \ dx &=& \bar{v}_i, \\
\int f_\phi(x)(x - \bar{v}_i)^2 \ dx &=& \phi^2.
\end{eqnarray}
\item[$\triangleright$] A non-linear threshold $\text{Th}_\theta$ is then applied to $\bar{v_i}$ to yield the neural state at the next time step:
\begin{equation}
  \text{Th}_{\theta_i} (S_i)= \begin{cases}
    0, & \text{if} \ \ v_i < \theta_i,\\
    1, & \text{otherwise}.
  \end{cases}
\end{equation}
\end{enumerate}
To summarize, the sequential transformation of the neuron state is:
\begin{eqnarray}
S^{(t+1)} &=& \text{Th}_\theta(X^{(t)}),\\
X^{(t)} & \sim & \mathcal{N}(CS^{t}|\phi),
\end{eqnarray}
where $\mathcal{N}(x|\phi)$ is a gaussian distribution with mean $x$ and standard deviation $\phi$. The algorithm is fully determined by the following set of transformation variables: $(C, f_\varphi, \theta)$.

\subsection{Determination of the Temperature factor}

In this section, we compute the stationary distribution of the spin variable $S^t$ for the general case with variables $(C, f_\varphi, \theta)$. We first determine the probability of a single spin-flip, knowing that the current spin state is $J$:
\begin{eqnarray}
W(\sigma_i (I) | J) &=& G(h_i (J) \sigma_i(I)), \\
h_i(J) & =& - \sum_{j=1}^N J_{ij}\sigma_j (J) - h_i^0, \\
h_i^0 &=& \sum_j C_{ij}/2 - \theta_i, \label{h0}
\end{eqnarray}
where we define $G$ as a rescaled version of the noise cumulative density function:
\begin{eqnarray}
F(x) &=& \int_{-\infty}^{x} f_\phi(u) \ du,\\
F_0(x) &=& F(x + \bar{v}_i), \\
G (x) & = & 1 - F_0(x).
\end{eqnarray}
In the following, we only assume that $f_\phi$ is symmetric and has a standard deviation $\phi$. The symmetry assumption is not constraining as most noise distributions found in nature have a symmetric (and, in many cases, gaussian) distribution \cite{Horowitz1990}. Our analysis can also be extended to noise distributions whose variance is infinite by parameterizing the distribution with another parameter, usually referred to as a "scaling" parameter (see Table \ref{temp_factor}). The case of $f_\phi$ being a gaussian distribution is discussed in Ref. \cite{Peretto1984}, where the function $G$ can be approximated by a sigmoid function when rescaling it by the adequate factor. In the general case, in order to minimize this approximation error, we choose the following rescaling factor
\begin{eqnarray*}
G_\gamma(x) & = & G(\gamma x),\\
1 - s(x) & = & \frac{1}{1 + e^x},\\
\gamma_\text{opt} &=& \text{argmin}_{\gamma'} \max_{X \in \mathbb{R}} | G_{\gamma'} (X) - (1-s(X)) | = \text{argmin}_{\gamma'} \max_{X \in \mathbb{R}} | \epsilon_{\gamma} (X)|,
\end{eqnarray*}
where $s(x) = 1 / (1 + e^{-x})$ is the sigmoid function.

Naturally, the optimal scaling factor will depend on the standard deviation $\phi$ of the original probability distribution $f$. We numerically observe that this dependence is linear
\begin{eqnarray}
\gamma_\text{opt} &=& \frac{k}{2} \ \phi.
\end{eqnarray}

Optimal scaling factors and corresponding errors are reported in Table \ref{temp_factor}. Identifying the transition probability to a sigmoid function is a necessary step to compute the effective Hamiltonian of the spin distribution \cite{Peretto1984}.

\begin{table}
\begin{center}
 \begin{tabular}{||c c c c||} 
 \hline
 Noise distribution & $G(x)$ & $\frac{k}{2}$ & $\epsilon_0 \equiv \max_{X \in \mathbb{R}} |\epsilon_{\gamma_\text{opt}}(X)|$ \\ [0.5ex] 
 \hline\hline
 Logistic & $\frac{1}{1+\exp\big(\frac{\pi x}{\sqrt{3} \phi}\big)}$ & $\frac{\sqrt{3}}{\pi}$ & 0 \\ 
 \hline
 Gaussian & $\frac{1}{2}(1-\text{erf}(\frac{x}{\sqrt{2}\phi}))$ & 0.5877 & 0.0095 \\
 \hline
 Cauchy$^{*}$ & $\frac{1}{\pi} \arctan (\frac{x}{\phi}) + \frac{1}{2}$ & 1.16 & 0.0495 \\
 \hline
 Laplace & $\left\{
\begin{array}{c l}	
     \frac{1}{2}\exp(-\frac{\sqrt{2}x}{\phi}) & x\leq 0\\
     1-\frac{1}{2}\exp(\frac{\sqrt{2}x}{\phi}) & x \geq 0
\end{array}\right.$ & 0.4735 & 0.0199 \\
 \hline
 Uniform & $\left\{
\begin{array}{c l}	
     0 & x\leq -\sqrt{3} \phi \\
     \frac{x+\sqrt{3}\phi}{2\sqrt{3}\phi} & x \in [-\sqrt{3}\phi, \sqrt{3}\phi[ \\
     1 & x \geq \sqrt{3}\phi
\end{array}\right.$ & 0.6136 & 0.0561 \\ [1ex] 
 \hline
\end{tabular}
\end{center}
\caption{\textbf{Summary of temperature factors for various noise distributions.} Each noise distribution is defined so that its standard deviation, when applicable, is equal to $\phi$. ($^*$the standard deviation of the Cauchy distribution is not defined, the linear dependence is measured as a function of the scaling parameter $\phi$ in this particular case.)}
 \label{temp_factor}
\end{table}

\subsection{Transition probability}
By choosing adequately the temperature factor $\frac{k}{2}$, we minimize the error when approximating the transition probability by a sigmoid function (in the following, $h_i = h_i (J)$ and $\sigma_i =\sigma_i (I)$, for readability):
\begin{eqnarray}
W(\sigma_i | J) &=& G_{\gamma_\text{opt}} \left( \frac{h_i \sigma_i}{\gamma_\text{opt}} \right) \\
 & = & s \left( \frac{h_i \sigma_i}{\gamma_\text{opt}} \right) + \epsilon_{\gamma_\text{opt}} \left( \frac{h_i \sigma_i}{\gamma_\text{opt}} \right) \\
 & \approx & s \left( \frac{h_i \sigma_i}{\gamma_\text{opt}} \right).
\end{eqnarray}
In the following, we will drop the subscript in $\gamma_\text{opt}$ and will assume that its value is known (it can be determined numerically, given the noise distribution). The error function $\epsilon$ can be used as an expansion parameter, in order to estimate the accuracy of the approximation performed when neglecting this term (as is done in Ref.\cite{Peretto1984}).

\subsubsection{Expansion of the transition probability in terms of $\epsilon$}
The transition probability can be derived by multiplying the probabilities of single spin-flips, these events being independent:
\begin{eqnarray}
W(I|J) &=& \prod_i W(\sigma_i |J) \\
 &=& \prod_i \left( s \left( \frac{h_i \sigma_i}{k\phi/2} \right) + \epsilon \left( \frac{h_i \sigma_i}{k\phi/2} \right) \right) \\
 &=& W^0 (I|J) + \sum_{k=1}^N W^k(I|J),
\end{eqnarray}
where $W^k(I|J)$ is the $k$-th order term in the expansion, assuming the error function $\varepsilon$ is small ($\varepsilon \ll s$). The zero-th order term of this expansion gives us back the result from \cite{Peretto1984}:
\begin{eqnarray}
W^0 (I|J) &=& \frac{e^{-\beta H^0 (I|J)}}{\sum_K e^{-\beta H^0 (K|J)}},\\
H^0 (I|J) & =& -\sum_{ij} J_{ij} \sigma_i(I) \sigma_j(J) -\sum_i h_i^0 \sigma_i (I),\\
\beta &=& \frac{1}{k\phi}.
\end{eqnarray}
Another way to write this zero-th order term \cite{Peretto1984} will prove to be useful later in our derivations:
\begin{equation}
W^0 (I|J) = \frac{e^{-\beta H^0 (I|J)}}{\prod_i 2 \cosh ( \beta h_i (J) )}.
\label{transprob2}
\end{equation}

The general expression for the $k-$th order term can be derived:
\begin{equation}
W^k (I|J) = W^0 (I|J) \sum_{j_1} \  ... \sum_{j_k \not\in \{ j_1, \ ..., \ j_{k-1} \} } \prod_l \epsilon \left( 2\beta h_{j_l} \sigma_{j_l} \right) \left( 1 + \exp \left( 2\beta h_{j_l} \sigma_{j_l} \right) \right).
\end{equation}
The $N-$th order term scales, in the worst case scenario, like $\epsilon_0^N$:
\begin{equation}
|W^N (I|J)| = \left|\prod_k \epsilon\left(2\beta h_{j_l} \sigma_{j_l}\right)\right| \leq \epsilon_0^N.
\end{equation}
For $k<N$, analyzing the scaling of the $k-$th order term requires more assumptions. Let's look at the case $k=1$, to verify that we can safely neglect higher-order terms in this expansion:
\begin{equation}
\left|\frac{W^1 (I|J)}{W^0 (I|J)} \right| \leq \sum_j \left| \epsilon\left(2\beta h_{j} \sigma_{j}\right) \left(1 + \exp\left({2\beta h_{j} \sigma_{j}}\right)\right) \right|.
\end{equation}
This ratio scales exponentially with $\beta h_{j} \sigma_{j}/2$. However, in this case, the error function $\epsilon$ also goes to zero. In addition, increasing the value of $\beta h_{j} \sigma_{j}/2$ also increases the value of the Hamiltonian $H^0 (I|J)$, and thus reduces the likelihood of the transition $W(I|J)$. \textbf{Thus, the larger the ratio $|\frac{W^1 (I|J)}{W^0 (I|J)}|$, the smaller the transition probability.} In the following, we will neglect all high order terms $k \geq 1$.

We do not extend our derivations to a general non-symmetric noise distribution in this work. We still suggest the following ideas to treat this extension:
\begin{enumerate}
\item[$\triangleright$] We could first treat the non-symmetric case as a perturbation of the symmetric case and thus derive the error on the effective Hamiltonian of the spin distribution.
\item[$\triangleright$] In this view, we suggest the parametrization of the skewness using the skew normal distribution \cite{Azzalini1996}, which is a natural extension of the toy-model symmetric noise distribution analyzed in Ref. \cite{Peretto1984}.
\end{enumerate}

\subsection{Detailed balance condition}
To determine the stationary distribution of the spin state, we first need to verify that the transition probability satisfies the triangular equality \cite{Peretto1984}, equivalent to the detailed balance condition in the context of Markov Chains \cite{Landau2009}:
\begin{equation}
W(I|K) W(K|J) W(J|I) = W(J|K) W(K|I) W(I|J) .
\end{equation}

This is equivalent to:
\begin{equation}
H(I|K) +  H(K|J) + H(J|I) = H(J|K) + H(K|I) + H(I|J),
\end{equation}
which is satisfied if $J$ is symmetric (there is no condition on the linear term $h^{0}_i$ in the Hamiltonian). We can then deduce the effective Hamiltonian by decoupling the ratio of transition probabilities into the ratio of two terms depending on distributions $I$ and $J$ (by using Eq. (\ref{transprob2})) \cite{Peretto1984}
\begin{eqnarray}
\frac{W^0 (I |J)}{W^0 (J | I)} & = & \frac{\exp (\beta \sum_i h^0_i \sigma_i(I) ) }{\exp (\beta \sum_i h^0_i \sigma_i(J) )} \frac{\prod_i \cosh (\beta (\sum_j J_{ij} \sigma_j(I) + h^0_i) }{\prod_i \cosh (\beta (\sum_j J_{ij} \sigma_j(J) + h^0_i) }\\
& = & \frac{F(I)}{F(J)},
\end{eqnarray}
and $H_L (I) = -\frac{1}{\beta} \ln F(I)$ gives the effective Hamiltonian describing the dynamics of the spin distribution:
\begin{equation}
H_L(I) = -\frac{1}{\beta} \sum_i \log \cosh (\beta (\sum_j J_{ij} \sigma_j (I) + h_i^0)) - \sum_i h_i^0 \sigma_i(I).
\end{equation}
\subsection{Effective Gibbs distribution}
Assuming $J$ is symmetric, the stationary distribution of the spin state is a Gibbs distribution given by the effective Hamiltonian $H_L (I)$\cite{Peretto1984, Amit1985}:
\begin{equation}
P (\{ \sigma_i \}) \propto \exp{ \left( -\beta H_L (\{ \sigma_i \}) \right)} .
\end{equation} 
In the following, we define $||.||_k$ as the usual $k$-norm (where $k$ is an integer). 
\subsubsection{Small noise approximation}
In the small noise approximation ($\beta \gg 1$), the effective Hamiltonian can be simplified using the following Taylor expansion : $\log \cosh x \approx |x| - \log2$:
\begin{eqnarray}
H_L (\{ \sigma_i \}, \beta ) & \approx & - \sum_{i} \big| \sum_j J_{ij} \sigma_j + h_i^0 |+\frac{N}{\beta} \log 2 - \sum_i h_i^0 \sigma_i\\
& = & - ||J \vec{\sigma} + \vec{h}^0||_1 + \frac{N}{\beta} \log 2  - \vec{h}^0 \cdot \vec{\sigma} := H_1 (\{ \sigma_i \}, \beta ).
\end{eqnarray}
\subsubsection{Large noise approximation}
In the large noise approximation ($\beta \ll 1$), the effective Hamiltonian can be simplified using the following Taylor expansion : $\log \cosh x \approx \log ( 1 + x^2 /2 ) \approx x^2 /2$,
\begin{eqnarray}
H_L(\{ \sigma_i \}, \beta ) &\approx & \underbrace{-\frac{\beta}{2} \sum_{ij} \tilde{K}_{ij} \sigma_i \sigma_j}_\text{Quadratic term} - \underbrace{\sum_i \sigma_i \left( \beta \sum_{k} J_{ki} h^0_k  + h_i^0 \right)}_\text{Linear term}  \\ 
&-& \underbrace{\frac{\beta}{2} \sum \left( h^0_j\right)^2}_\text{Constant (independent of $\sigma$)}  \\
&=& - \frac{\beta}{2}||J \vec{\sigma} + \vec{h}^0 ||_2^2 - \vec{h}^0 \cdot \vec{\sigma}  := H_{2} (\{ \sigma_i \}, \beta), \label{largenoise_exp_h}
\end{eqnarray}
with $\tilde{K} = J^2$.

\subsubsection{Discussion on $h^0_i$}

Identifying the Hamiltonian in Eq. (\ref{largenoise_exp_h}) to the general Ising Hamiltonian in Eq. (\ref{ising_general}) requires solving the following system of equations:
\begin{empheq}[left=\empheqlbrace]{align}
K_{ij} + \delta_{ij} \Delta_{i} &= \beta \sum_k J_{ik} J_{kj}  \label{system1}\\
b_i &= \beta \sum_j J_{ki} h^0_k + h^0_i \label{system2}\\
J_{ij} &= J_{ji},
\end{empheq}
with unknown variables $(h^0_i, J)$. $\delta_{ij}$ is the Kronecker symbol and $\Delta_{i}$ represents a degree of freedom on the diagonal of the Hamiltonian (which results in an additional constant term). Studying the set of general Ising Hamiltonians $(K,b)$ (Eq. (\ref{ising_general})) which can be mapped to an Ising network $(J,h^0)$ in the large noise approximation (Eq. (\ref{largenoise_exp_h})) would require an extensive study that we do not carry in this work. However, one can notice that:
\begin{enumerate}
\item[$\triangleright$] In the case where the external field satisfies $J \cdot h^0 = 0$, the system has a trivial solution $h^0_i = b_i$ and $J = \sqrt{K}$ (the conditions for which a square root of $K$ can be found are discussed in the next section).
\item[$\triangleright$] In the more general case, one is given $N$ "free" degrees of freedoms $\Delta_i$ in finding solutions to Eq. (\ref{system1}), while there are $N$ constraints to solve in Eq. (\ref{system2}). Thus, it seems likely that, in non-degenerate cases, this system will have a non-trivial solution. However, one should make sure that while tuning the degrees of freedom $\Delta_i$, the matrix $J$ remains a real square root of the matrix $K$ (see discussion in the next section on finding a square root). If not, this algorithm should be generalized to using complex-valued matrices $J$. We do not discuss this generalization of the algorithm in this work. However, considering that arbitrary (complex) unitary transformations can be implemented with our photonic architecture \cite{Shen2017a}, this would be an interesting extension to this work.
\end{enumerate}

\textbf{In the following, for the sake of simplicity, we will assume $h^0_i \equiv 0$.} In this case, the large noise approximation Hamiltonian gives: 
\begin{equation}
H_L(\{ \sigma_i \}, \beta) \approx -\frac{\beta}{2} \sum_{ij} \tilde{K}_{ij} \sigma_i \sigma_j ,
\end{equation}
with $\tilde{K} = J^2$. Thus, the implementation of the Little network to model a general Ising Hamiltonian (Eq. (\ref{ising_general})) strongly relies \textbf{on the possibility of finding a symmetric, real square root to the matrix $K$.}

Let us remind the reader of the assumptions we have made so far: 
\begin{enumerate}
\item[$\triangleright$] The model is \textit{synchronous}, or as stated in Little's original paper \cite{Little1974} "we shall suppose that the neurons are not permitted to fire at any random time but rather that they are synchronized such that they can only fire at some integral multiple of a period $\tau$". 
\item[$\triangleright$] The neurons have no memory of states older than their previous state (Markov process).
\item[$\triangleright$] The connections do not evolve in time (there is no learning involved).
\item[$\triangleright$] $J$ is symmetric and $h^0_i = 0$, which results in
\begin{equation}
\theta_i = \sum_j C_{ij}/2,
\end{equation}
from Eq. (\ref{h0}).
\end{enumerate}

\subsubsection{Inequalities between Hamiltonians}
Using basic algebra and functional analysis, we can show that $\log \cosh x \leq x^2/2$ and that $\log \cosh x \geq |x| - \log 2$, which results in the following inequality on the Hamiltonians
\begin{equation}
H_2 (\{ \sigma_i \}, \beta ) \leq H_L (\{ \sigma_i \}, \beta) \leq H_1 (\{ \sigma_i \}, \beta).
\end{equation}
By summing over spin configurations, we get the reversed inequality for the partition functions
\begin{equation}
Z_2 (\beta ) \geq Z_L(\beta) \geq Z_1(\beta).
\end{equation}

Another inequality can be derived by using the equivalence of norms $||.||_1$ and $||.||_2$ in finite dimensions ($ ||.||_2 \leq ||.||_1 \leq \sqrt{N} ||.||_2$):
\begin{equation}
0 \leq \sqrt{ - \frac{2}{\beta} H_2 (\{ \sigma_i \}, \beta )} \leq - H_1 (\{ \sigma_i \}, \beta)  + \frac{N}{\beta} \log 2 \leq \sqrt{\frac{-2N}{\beta}  H_2 (\{ \sigma_i \}) } .
\end{equation}

\begin{figure}
\begin{center}
\includegraphics[scale=0.55]{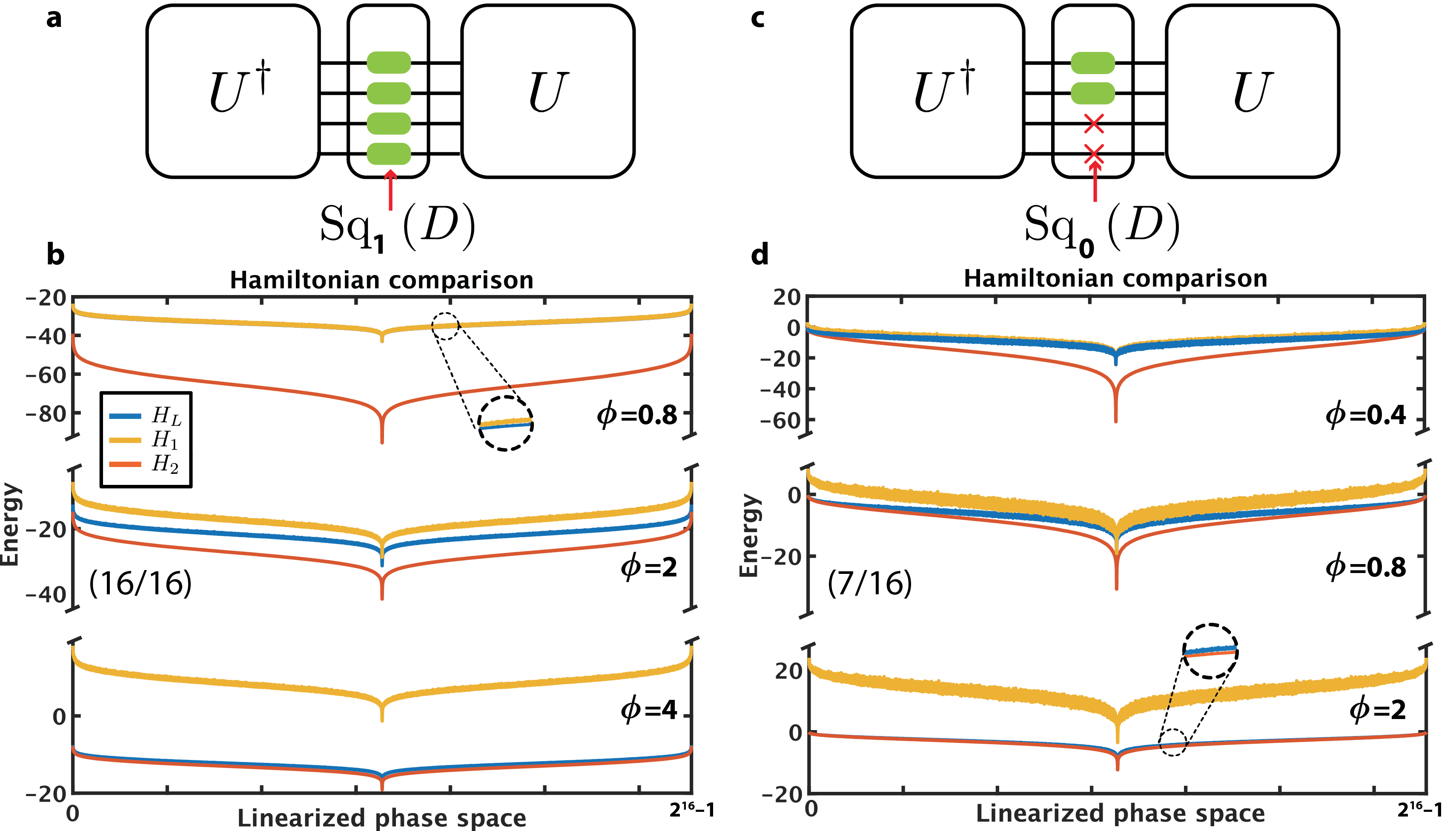}
\caption{\textbf{Transitioning between Hamiltonians. a:} To approximate a given Ising model in the large noise limit, one should take into account all eigenvalues. \textbf{b:} Effective Hamiltonian and its large/low-noise limit approximation when $\alpha = 1$ (with (16/16) eigenvalues). \textbf{c:} However, to reduce the noise threshold of the large noise expansion, one should drop out some eigenvalues. \textbf{d:} Effective Hamiltonian and its large/low-noise limit approximation when $\alpha = 0$ (with (7/16) eigenvalues, all negative eigenvalues being dropped out). In this Figure, we study a random spin glass with $N=16$ and choose $\Delta_{ii} = |\sum_{j} K_{ij}|$.}
\label{hamiltonians}
\end{center}
\end{figure}

These Hamiltonians are plotted for a sample random Ising problem in Figure \ref{hamiltonians}.

\subsubsection{Suggested algorithms}
As suggested above, from the large noise expansion of the effective Hamiltonian, a natural way to probe the Gibbs distribution of some Ising model defined by coupling matrix $K$ (Eq. (\ref{ising_general}) is to set the matrix of the recurrent loop $C$ to be a modified square root of $K$:
\begin{equation}
\text{Sq}_\alpha (J) = \text{Re} ( \sqrt{K + \alpha \Delta} ),
\label{sq_cond}
\end{equation}
where $\alpha \in [-1,1]$ is the eigenvalue dropout level, a parameter we will discuss below, and $\Delta$ is a diagonal offset matrix defined as the sum of the off-diagonal term of the Ising coupling matrix $\Delta_{ii} = \sum_{j \neq i} |K_{ij}|$. An alternative solution would be to set $C = K$. This solution has been studied in the particular case of associative memory couplings \cite{Amit1985}, where it was shown that the PRIS in this configuration would be described by a free energy function equal to the Ising free energy (up to a factor of 2). We will discuss the pros and cons of both algorithms in section IV.

\newpage

\section{Supplementary Note 3: Construction of the weight matrix}
We here propose a technique to build a square root of $K$ in order to find an approximate solution to \textit{any} Ising problem. We start with the general $K$ Ising weight matrix $K$ from Eq. (\ref{ising_general}). We notice that as $K$ is symmetric and $K_{ii} = 0$, $K$ will never obey the condition of Lemma 1. 

\begin{theorem}{(Diagonal dominance.)}
If a $A \in S_N (\mathbb{R})$ is a real symmetric matrix such that for all $i$, $|A_{ii}| \geq \sum_{j \neq i} |A_{ij}|$ and $A_{ii} > 0$, there exists $B \in S_{N} (\mathbb{R})$ such that $B^2 = A$.
\label{Lemma1}
\end{theorem}
\begin{proof}
Let $X$ be a vector of size $N$ and norm $1$ such that $AX = \lambda X$ with $\lambda < 0$. We assume $i_0$ is the coordinate of $X$ with maximum absolute value (and $x_{i_0} > 0$): $x_{i_0} = ||x||_{\infty}$. We have $\sum_{j} A_{i_0 j} x_{j} = \lambda x_{i_0}$. We thus get :
\begin{eqnarray*}
||x||_{\infty} |A_{i_0 i_0}| &<& |A_{i_0 i_0} - \lambda | ||x||_{\infty}\\
&=& | \sum_{j\neq i_0} A_{i_0 j} x_j |\\
& \leq & \sum_{j\neq i_0} |A_{i_0 j}| ||x||_{\infty}.
\end{eqnarray*}

By simplifying this inequality, we get $|A_{i_0 i_0}| < \sum_{j\neq i_0} |A_{i_0 j}|$ which contradicts our assumption. Thus, for all $X$, $X^T AX \geq 0$, which means that $A$ is positive semidefinite. 

We can thus write $A=UDU^T$ where $U$ is unitary and $D = \text{Diag}(\lambda_1,\cdots,\lambda_N)$ with $\lambda_i \geq 0$. The result is given by $B = U \sqrt{D} U^T \in S_N (\mathbb{R})$ with $\sqrt{D} = \text{Diag}(\sqrt{\lambda_1},\cdots,\sqrt{\lambda_N})$ 
\end{proof}

To be able to apply Lemma \ref{Lemma1}, we need to add diagonal terms to the Ising matrix $K$ from Eq. \ref{ising_general} by defining $\tilde{K} = K + \Delta$ where $\Delta$ is a diagonal matrix such that $\tilde{K}$ verifies the assumptions of theorem 1. From there, we can define $J = \sqrt{\tilde{K}}$. In the large noise approximation, we thus get the following Hamiltonian describing the PRIS: 
\begin{eqnarray*}
H_L(\{ \sigma_i \}, \beta) \approx -\frac{\beta}{2} \sum_{ij} \tilde{K}_{ij} \sigma_i \sigma_j &=& -\beta \sum_{1 \leq i < j \leq N} \tilde{K}_{ij} \sigma_i \sigma_j - \frac{\beta}{2} \sum_{1 \leq i \leq N} \tilde{K}_{ii} \sigma_i^2 \\
&=& -\beta \sum_{1 \leq i < j \leq N} K \sigma_i \sigma_j - \frac{\beta}{2} \sum_{1 \leq i \leq N} \Delta_{ii}.
\end{eqnarray*}
 Since the $\Delta$ matrix is fixed, the second term in this equation is a constant (independent of the spin distribution, given an Ising problem to minimize). The Gibbs distribution is thus $P(\{\sigma_i\}) \propto \exp{ \left( - \beta H_L(\{ \sigma_i \}) \right)} \approx C_{te} \exp{ \left( - \beta H_{K,0} \right)}$ where $H_{K,0}$ is the Ising Hamiltonian defined in Eq. (\ref{ising_general}).
 
\subsection{Discussion on the diagonal offset} The diagonal offset that is added to the original Ising matrix can be considered as a supplementary degree of freedom. To verify Lemma 1, we can choose $\Delta_{ii} = \sum_{j \neq i} | K_{ij}|$, to make sure $\tilde{K}$ is positive definite and has a real symmetric square root $J^2 = \tilde{K}$.

However, there is no reciprocal to Lemma 1. We here show that even some partial or weak reciprocals are usually wrong:
\begin{enumerate}
\item[$\triangleright$] The direct reciprocal of Lemma 1 is generally wrong. We consider the following counter-example:
\[ A =  \left( \begin{array}{cc}
1 & -2  \\
-2 & 8 \end{array} \right) \].
A is symmetric, has positive diagonal elements but is not diagonally dominant (the inequality does not hold for only one diagonal element). However, $A$ is symmetric positive definite: its determinant is 4 and its trace 9, so the sum and product of eigenvalues is positive, which means both are positive.
\item[$\triangleright$] Even a weak reciprocal of Lemma 1 is wrong:
\begin{center}
If $A \in S_N (\mathbb{R})$, $A$ has positive diagonal elements, such that for all $i$, $|A_{ii}| < \sum_{j \neq i} |A_{ij}|$, then $A$ has at least one negative eigenvalue.
\end{center}
For example, one can consider :
\[ A =  \left( \begin{array}{ccc}
1 & 1 & 1  \\
1 & 1 & 1 \\
1 & 1 & 1  \end{array} \right) \].
$A$ is symmetric, has positive diagonal elements, verifies: $\forall i$, $|A_{ii}| < \sum_{j \neq i} |A_{ij}|$, but still, its eigenvalues are 0, 0, and 3.
\item[$\triangleright$] One (very) weak reciprocal of Lemma \ref{Lemma1} is the following Lemma \ref{Lemma2}. It means that any Ising matrix (that has a zero diagonal) will never have a real symmetric square root, unless we add an offset to its diagonal. 
\begin{theorem}
If $A \in S_N (\mathbb{R})$ (real symmetric matrix) such that for all $i$, $A_{ii} = 0$, then $A$ has at least one negative (and one positive) eigenvalue.
\label{Lemma2}
\end{theorem}
\begin{proof}
We have $\sum_{i} A_{ii} = \sum_{i} \lambda_i = 0$ so there is at least one negative and one positive eigenvalue.
\end{proof}
\item[$\triangleright$] \begin{theorem}[Weak reciprocal: unweighted MAX-CUT case.] If $K$ is equal to minus the adjacency matrix of a graph, then the matrix $\tilde{K} = K + \alpha \Delta$ is symmetric positive semi-definite \textbf{if and only if} $\alpha = 1$.
\label{Lemma3}
\end{theorem}
\begin{proof}
Let's first notice that this case still holds for a NP-hard subclass of Ising problems. In this case, $\tilde{K}$'s elements are $-1$ if there is a spin-spin connection and 0 otherwise. Thus, $\tilde{K} = K + \Delta$ has a zero eigenvalue: $KX = 0$ with $X = (1,1,...,1)^T$ (because the sum of each row is equal to 0). We notice that for $\alpha < 1$:
\begin{equation}
K + \alpha \Delta = (K + \Delta) + (\alpha - 1) \Delta.
\end{equation}
As seen before, $K + \Delta$ has a zero eigenvalue, and $(\alpha - 1) \Delta$ only has negative eigenvalues. Using Weyl's inequalities, we can deduce that $K + \alpha \Delta$ has at least one strictly negative eigenvalue. 
\end{proof}
\end{enumerate}
Studying the reciprocal of Lemma 1 gives us a better insight on what is the optimal diagonal offset $\Delta$ to add to the original Ising matrix. We can consider $\Delta$ as a supplementary degree of freedom in our algorithm. The choice of $\Delta$ suggested by this Lemma on diagonal dominance is particularly adapted for $\alpha>0$ and Ising problems with couplings that all have the same sign. However, we notice that for large spin glasses, typically $\sum_{j\neq i} |K_{ij}|$ will be much greater than $K_{ii}$. Thus, $\alpha \sim - 1/N$ is the lower limit for non-zero $J$ matrices. The domain of definition of $\alpha$ over which the number of eigenvalues actually varies is then asymmetric. For this reason, we typically choose $\Delta_{ii} = | \sum_{j \neq i} K_{ij}|$ for large spin glasses. In this study, we will specify which definition of $\Delta$ is chosen for each Figure and study.

\subsection{Tuning the search dimensionality with eigenvalue dropout}
Tuning the eigenvalue dropout gives us a way of tuning the search dimensionality. If $K = UDU^\dagger$ where $U$ is real unitary and $D = \text{Diag}(\lambda_1, ... \lambda_N)$, we can parametrize the phase space in terms of the eigenvectors of $K$, which we denote as $e^j$ (associated to eigenvalue $\lambda_j$): 
\begin{eqnarray}
\vec{\sigma} &=& \sum_j \mu_j e^j, \\
H(\vec{\sigma}) &=& - \frac{1}{2} \sigma^T K \sigma \\
& = & - \frac{1}{2} \sum_j \mu_j^2 \lambda_j.
\end{eqnarray}
We can rephrase the optimization problem as follows:
\begin{equation}
\text{Find} \ \ \ 
\begin{cases}
\text{argmin}_\mu -\frac{1}{2}\sum_j \mu_j^2 \lambda_j \\
\sum_j \mu_j e^j_i = \pm 1
\end{cases} \label{system_opt}
\end{equation}

As only the positive eigenvalues can decrease the energy, we would like to conclude that only the eigenvectors associated with positive eigenvalues will contribute to minimizing the energy. However, one should make sure that the hard constraint in Eq. (\ref{system_opt}) remains satisfied. We can only conclude on the heuristic result that the ground state of the optimization probably will prefer having components in eigenspaces associated with positive eigenvalues. As reducing the eigenvalue dropout level results in dropping out the negative eigenvalues, this explains why we usually observe a better performance for a certain level around $\alpha = 0$.

\subsection{Modified algorithm with tunable offset} 
We define $\Delta^{0}$ as the minimum offset to verify the assumption of Lemma 1 and make sure the modified Ising matrix is symmetric positive semi-definite. $\Delta^{0}$ is defined as:
\begin{eqnarray*}
\Delta^{0}_{ii} &=& \sum_{j\neq i} |A_{ij}|, \\
\Delta^{0}_{ij} &=& 0 \ \ \ \text{for $i \neq j$}.
\end{eqnarray*}
We define the modified algorithm as follows:
\begin{enumerate}
\item[$\triangleright$] $K = \tilde{K} + \alpha \Delta^{0}$ where $\alpha \in [0,1]$
\item[$\triangleright$] $J =  \text{Re} \sqrt{K}$, defined as follows: as $K$ is symmetric, there is a unitary (and real) $U$ such that $K = U \text{Diag}(\lambda_1, ..., \lambda_N) U^{\dagger}$. We define $\sqrt{\lambda_i} = i\sqrt{|\lambda_i|}$ if $\lambda_i < 0$. Then
\[J = U \Big( \text{Diag}(\sqrt{\lambda_1}, ..., \sqrt{\lambda_N}) \Big) U^{\dagger}.\]
\end{enumerate}
This modified algorithm only corresponds to a particular parametrization of the weight matrix and can thus be implemented with the PRIS. The writing of this problem as a function of the coupling matrix eigenvalues (see above) proves that, as long as $\alpha \geq 0$ in the modified algorithm, the ground state of the Ising problem will likely remain unchanged (see Figure \ref{dropout_SI}).

\newpage
\section{Supplementary Note 4: Numerical simulations}

\begin{figure}
\begin{center}
\includegraphics[scale=0.55]{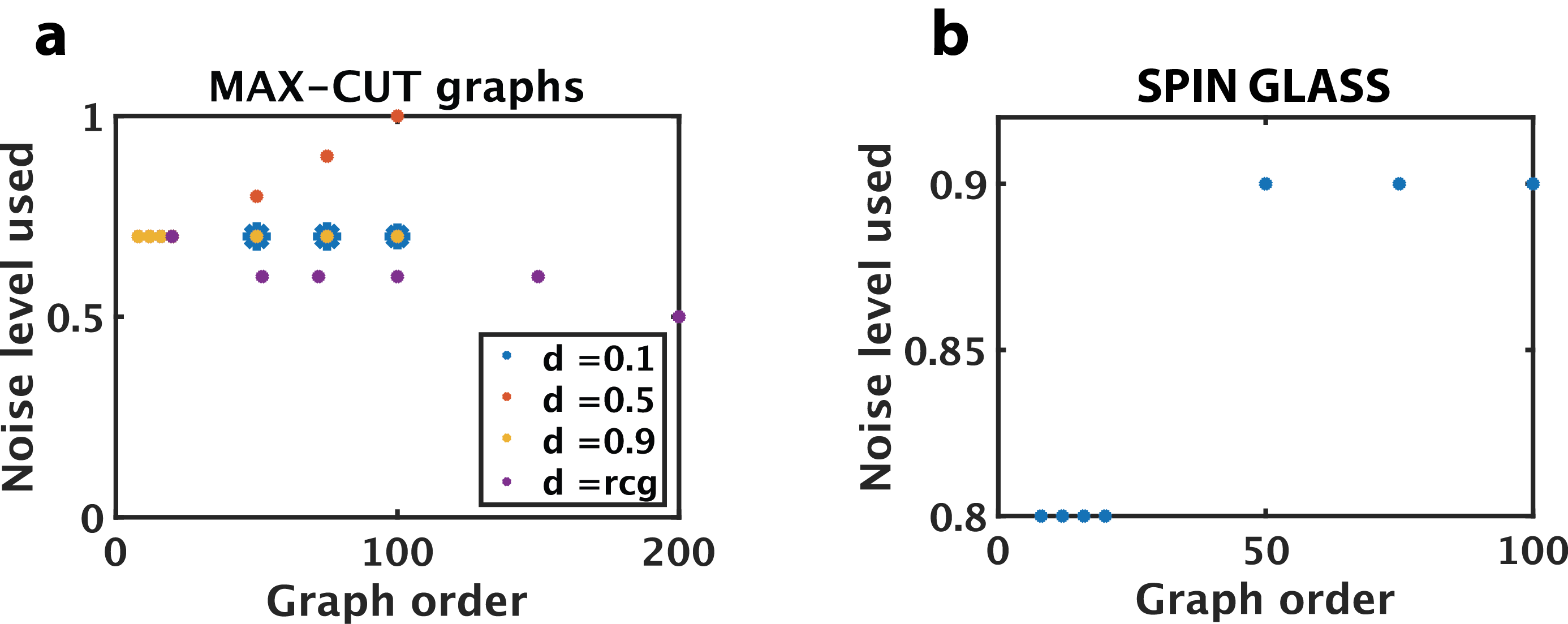}
\caption{\textbf{Optimized noise level used in Figure 2 (main text).} \textbf{a:} for MAX-CUT graphs (Figure 2c in the main text). \textbf{b:} for spin glasses (Figure 2b in the main text). }
\label{phi_data_fig3}
\end{center}
\end{figure}

\begin{figure}
\begin{center}
\includegraphics[scale=0.6]{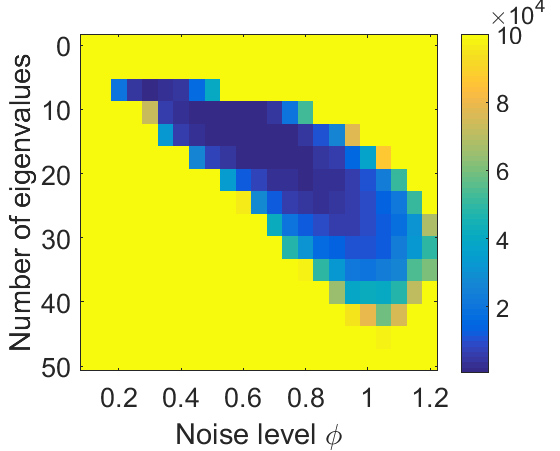}
\caption{\textbf{Conjugated influence of eigenvalue dropout and noise levels for finding the ground state.} For each tuple ($\alpha$, $\phi$), the PRIS is run 100 times and $N_{\text{iter}, 90\%}$ is plotted (the ground state is pre-computed with BiqMac \cite{Rendl2010}). The colorbar is saturated at $10^5$ iterations in order to show the valley of optimal $(\alpha, \phi)$ values. In this Figure, we study a MAX-CUT graph with density $d=0.5$ and $\Delta_{ii} = \sum_{j \neq i} |K_{ij}|$.}
\label{2dsweep_dropout_phi}
\end{center}
\end{figure}

\begin{figure}
\begin{center}
\includegraphics[scale=0.5]{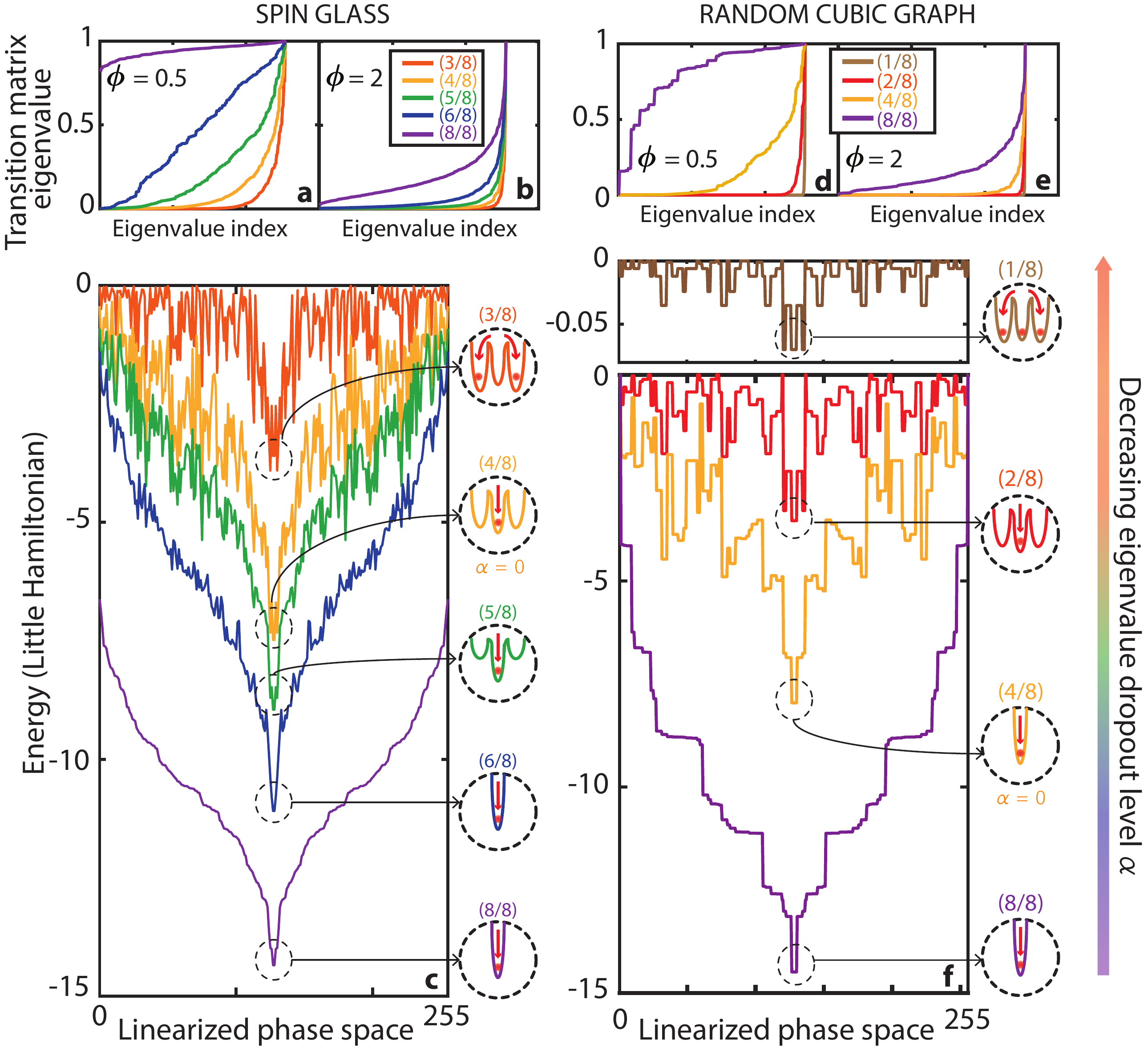}
\caption{\textbf{Influence of dropout} (a-c) for a random spin glass and (d-f) for a sample random cubic graph with $N = 8$ and $\phi = 0.5$. (a, b, d, e): Eigenvalue distribution for various noise levels. (c, f): Energy landscape (Hamiltonian) plotted over linearized phase space for various eigenvalue dropout levels $\alpha$. For this Figure, we set $\Delta_{ii} = \sum_{j} |K_{ij}|$.}
\label{dropout_SI}
\end{center}
\end{figure}

\subsection{Evaluating sampling performance and critical parameters on two-dimensional ferromagnetic and infinite-range Ising models}

\begin{figure}
\begin{center}
\includegraphics[scale=0.5]{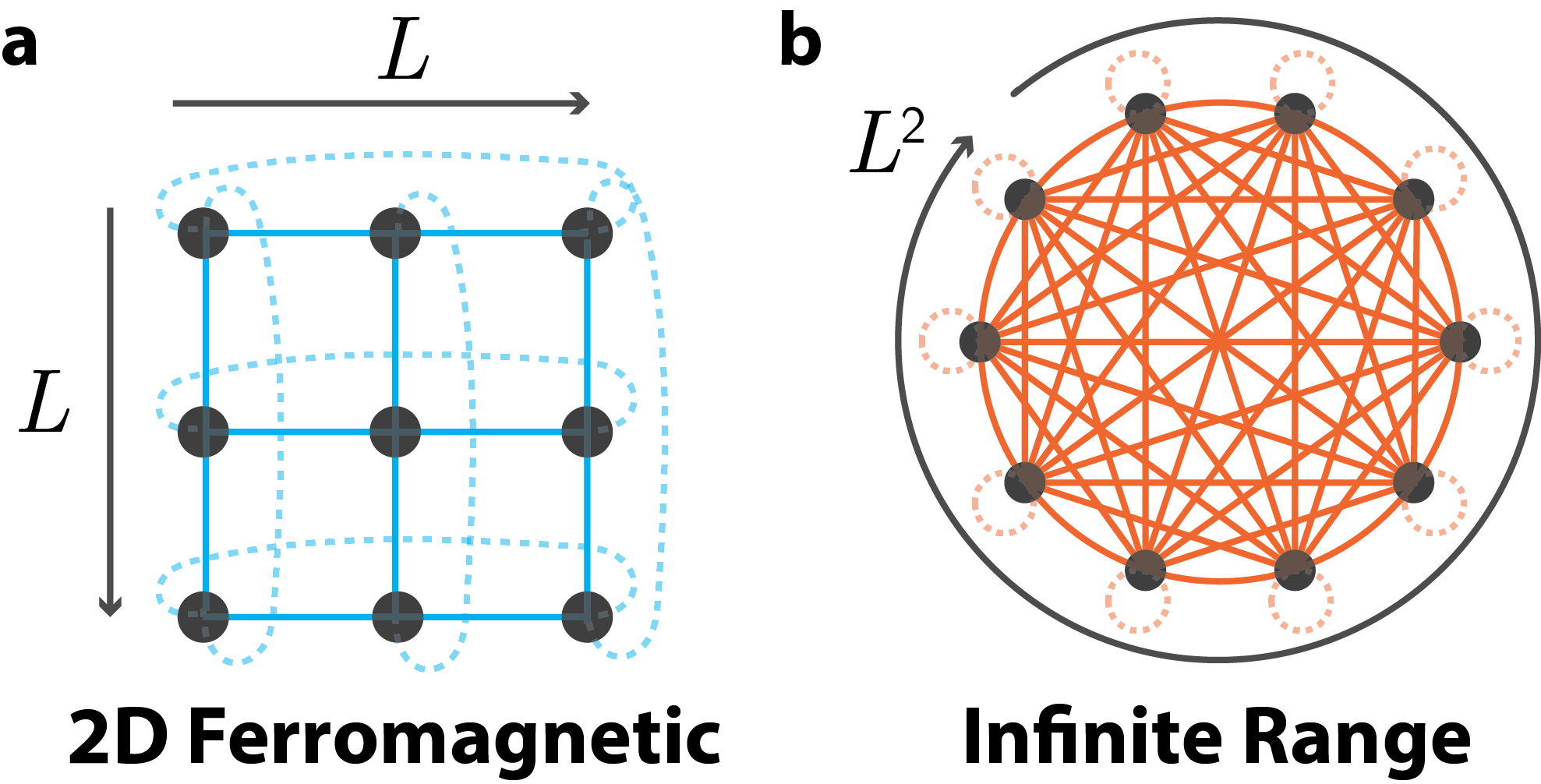}
\caption{\textbf{Examples of analytically solvable Ising models studied in our benchmark.} (a) Two-dimensional ferromagnetic Ising model and (b) infinite-range (mean-field) Ising model.}
\label{benchmark_examples}
\end{center}
\end{figure}

In the main text, we evaluate the performance of the PRIS on sampling the Gibbs distribution of analytically solvable Ising models and estimating their critical exponents. We chose the following two problems:
\begin{enumerate}
\item[$\triangleright$] Two-dimensional ferromagnetic Ising model on a square lattice with periodic boundary conditions. This model was first analytically solved by Onsager \cite{Onsager1944}. Its energy is defined as 
\begin{equation}
H = - \frac{1}{2} \sum_{\langle ij \rangle } \sigma_i \sigma_j ,
\end{equation}
where $\langle ij \rangle$ refers to nearest neighbors $(i,j)$ under periodic boundary conditions (see Figure \ref{benchmark_examples}).
\item[$\triangleright$] Infinite-range Ising model, where each node is connected with a positive coupling of $1/N$ to all its neighbors, including itself. This model can be analytically solved by mean field theory \cite{Friedli2017}. Its energy is defined as 
\begin{equation}
H = - \frac{1}{2N} \sum_{ij} \sigma_i \sigma_j = - \frac{1}{2N} (\sum_{i} \sigma_i)^2.
\end{equation}
\end{enumerate}

\begin{figure}
\begin{center}
\includegraphics[scale=0.5]{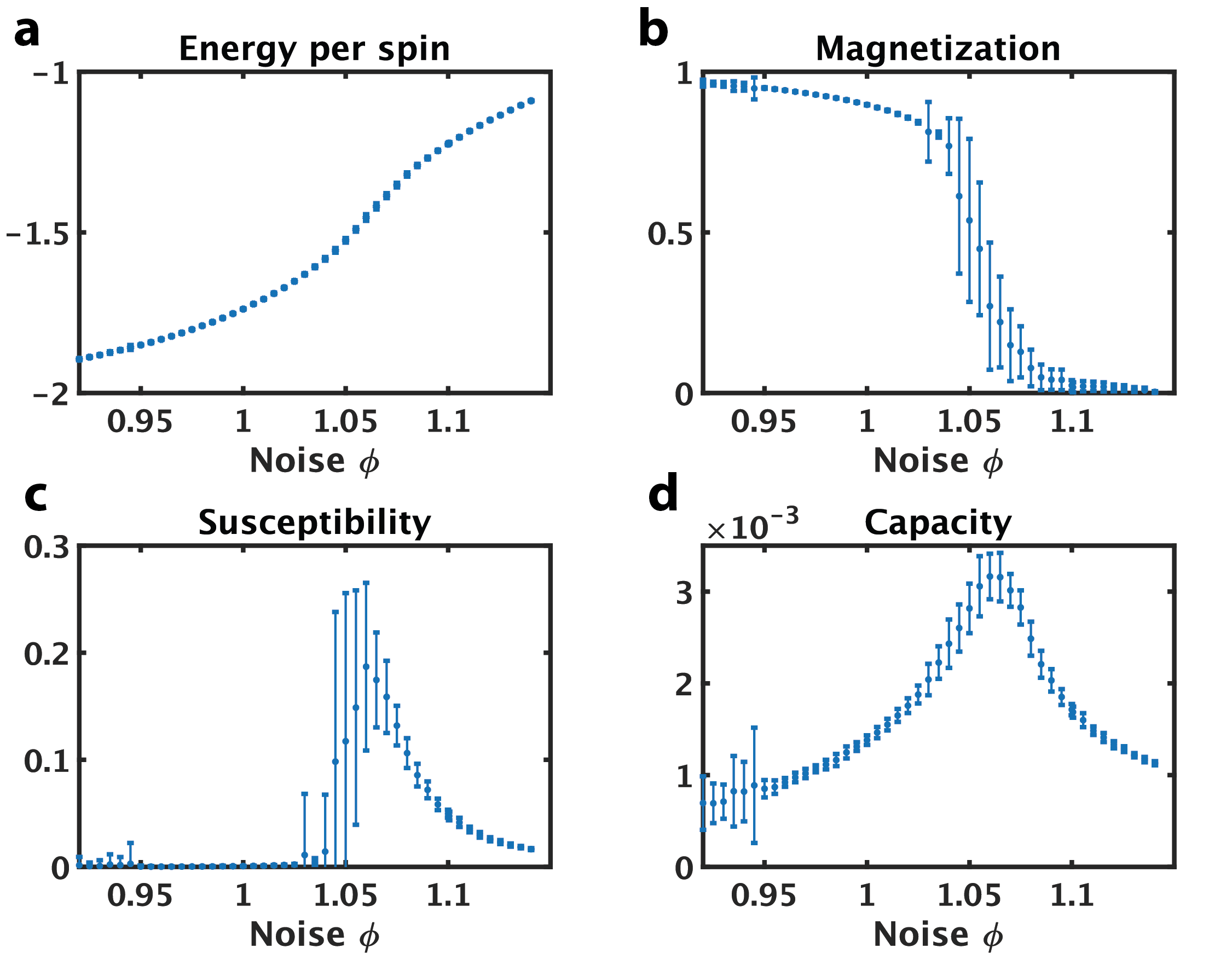}
\caption{\textbf{Physical observables of the two-dimensional Ising model obtained with $L = 36$ square-rooted PRIS.}}
\label{36by36_sq}
\end{center}
\end{figure}

In order to measure the critical exponents, we need to make sure the free energy at all temperatures is equivalent to that of the corresponding Ising model \textit{at all temperatures.}  It has been shown that the free energy of a Little network with coupling matrix $K$ equals that of an Ising model defined by the same coupling matrix $K$, when the coupling weights are configured to learn a set of stable configurations (associative memory) \cite{Amit1985}. However, this network has been discarded because of the possible existence of loops between some of its degenerate states, which can result in odd dynamics of the system, if no additional precautions are taken (see, for instance, Figure \ref{nosq_alg}, where this simple algorithm actually converges to the \textit{maximal} energy). We observe that adding a diagonal offset to the coupling, as we do in the square-root version of this algorithm, suppresses these odd dynamics (Figure \ref{nosq_alg}(b)).

The algorithm described earlier to find the ground state of Ising problems could also be used to measure critical exponents, as can be seen in Figure \ref{36by36_sq}. However, there are some complications that arise:
\begin{enumerate}
\item[$\triangleright$] Taking the square root of the coupling matrix prevents the use of symmetry and sparsity to reduce the algorithm complexity. Thus, for large graphs, the time needed to make a single matrix multiplication becomes quite large on a CPU.
\item[$\triangleright$] Taking the square root also modifies the coupling amplitude. It is now unclear how the temperature of the system should be defined (which is a problem if we want to estimate the critical temperature with the PRIS). From the large-noise expansion of the Hamiltonian, we should define the effective temperature as $T = k^2 \phi^2$. However, we observe this definition does not match the theoretical value of the critical temperature of the 2D Ferromagnetic Ising model. 
\end{enumerate}

\begin{figure}
\begin{center}
\includegraphics[scale=0.5]{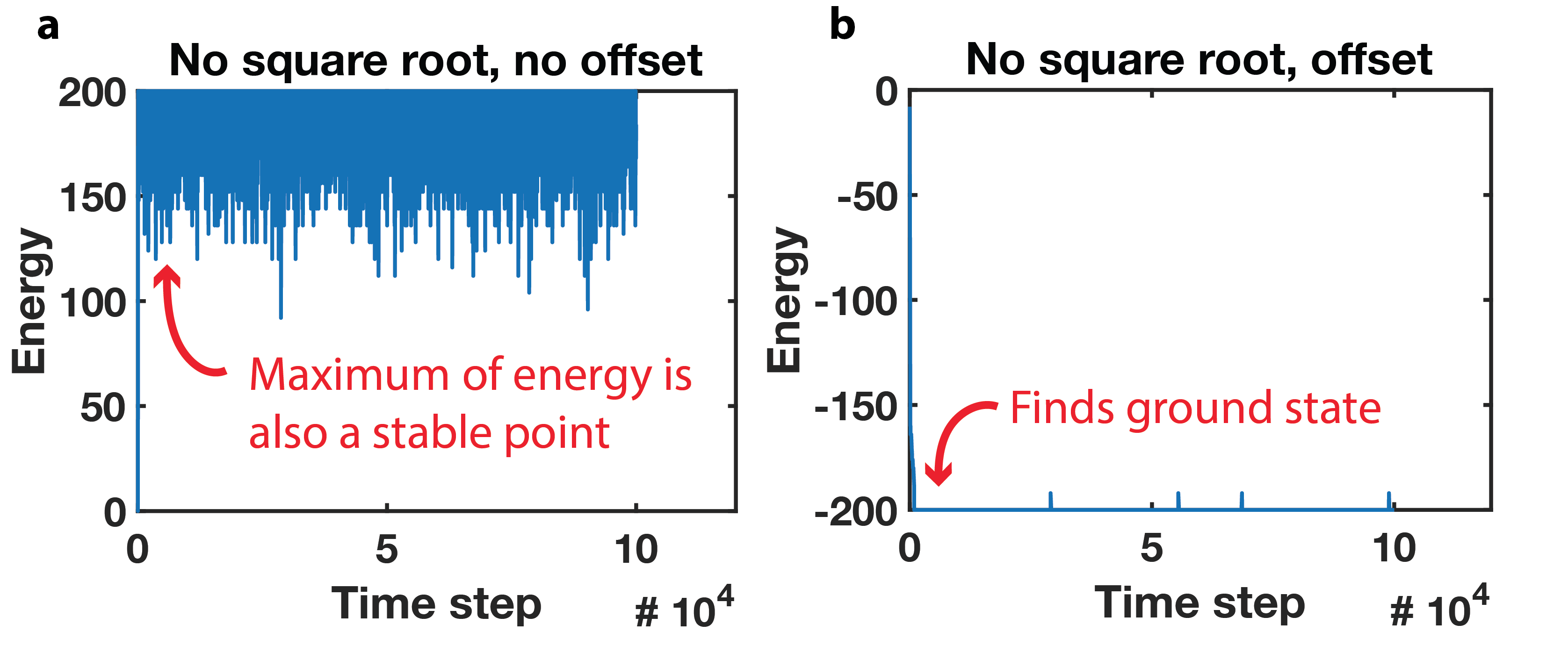}
\caption{\textbf{Comparison of different versions of the algorithms on the 2D Ferromagnetic Ising Model.} (a) With no offset, the no square root algorithm (defined by the following regime of operation of the PRIS: $C = K$) results in convergence to the state with maximal energy, for some runs (with random initial conditions). (b) Adding an offset to the same regime of operation cancels this behavior and the algorithm always converges to the ground state after some time.}
\label{nosq_alg}
\end{center}
\end{figure}

For both Ising problems studied, we perform the following analysis:
\begin{enumerate}
\item[$\triangleright$] We first estimate the Binder cumulant $U_4$ (see definition in the main text) as a function of the system temperature, for various graph sizes $N = L^2$. The cumulants $U_4$ plotted for different $L$ intersect at the critical temperature \cite{Landau2009} (see Figure 4(a) in the main text).
\item[$\triangleright$] We then estimate the dependence on the linear dimension $L$ of various observables at the critical temperature and deduce the corresponding critical exponent. This is enabled by the scaling law of observables at the critical temperature \cite{Landau2009}:
\begin{eqnarray}
m & \sim & L^{-\beta_C/\nu_C}, \\ 
\chi & \sim & L^{\gamma_C/\nu_C}, \\
\tau_\text{auto}^E & \sim & L^{z^E_C}, \\
\tau_\text{auto}^m & \sim & L^{z^m_C},
\end{eqnarray}
where $m, \chi, \tau_\text{auto}^E,$ and $\tau_\text{auto}^m$ are respectively the magnetization, the magnetic susceptibility, the energy autocorrelation time and the magnetization autocorrelation time \cite{Landau2009}. To estimate the critical exponents, we run the algorithm 120 times for $10^5$ iterations, with random initial conditions, at the critical temperature for each $L$ and fit these observables with a power law in $L$.
\end{enumerate}
The estimates of all observables are obtained by taking every $\tau^E_\text{auto}$ generated samples, where $\tau^E_\text{auto}$ is the energy autocorrelation time, and dropping the first $10\%$ of the data (arbitrary burn-in or equilibrium time). Our findings are summarized in Table \ref{criticalexp_table}, Figures \ref{fits_2d_ferro} and \ref{fits_infinite_range}, and are benchmarked versus the Metropolis-Hastings algorithm, which is summarized in Refs.\cite{Landau2009, Metropolis1953, Hastings1970}.

\begin{figure}[h!]
\begin{center}
\includegraphics[scale=0.75]{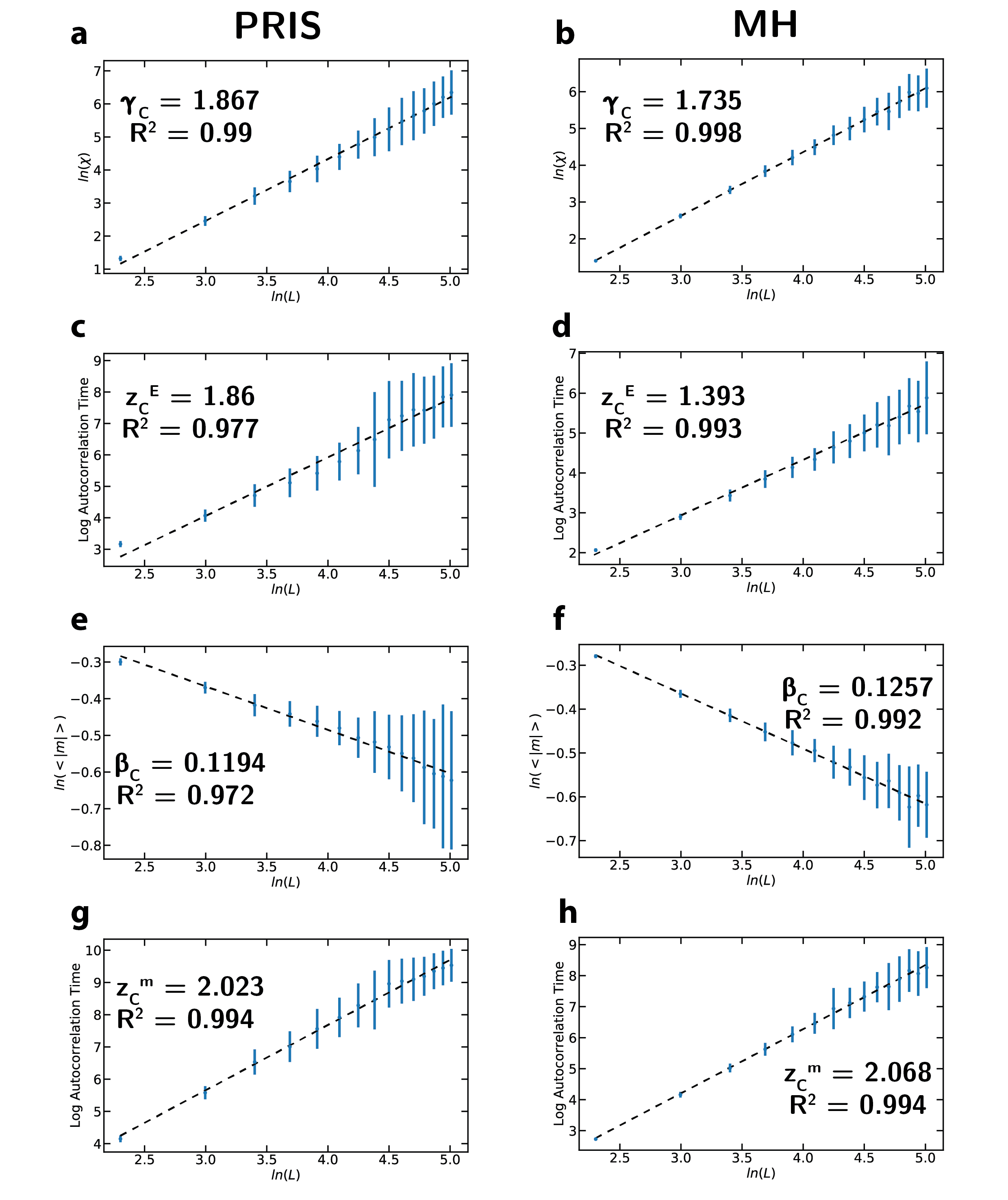}
\caption{\textbf{Probing the critical exponents of the 2D Ferromagnetic Ising model.} Fits are shown with the resulting critical exponent for the PRIS \textbf{(a, c, e, g)} and the MH \textbf{(b, d, f, h)} algorithms for the susceptibility \textbf{(a-b)}, energy autocorrelation time \textbf{(c-d)}, magnetization \textbf{(e-f)}, and magnetization autocorrelation time \textbf{(g-h)}.}
\label{fits_2d_ferro}
\end{center}
\end{figure}

\begin{figure}[h!]
\begin{center}
\includegraphics[scale=0.75]{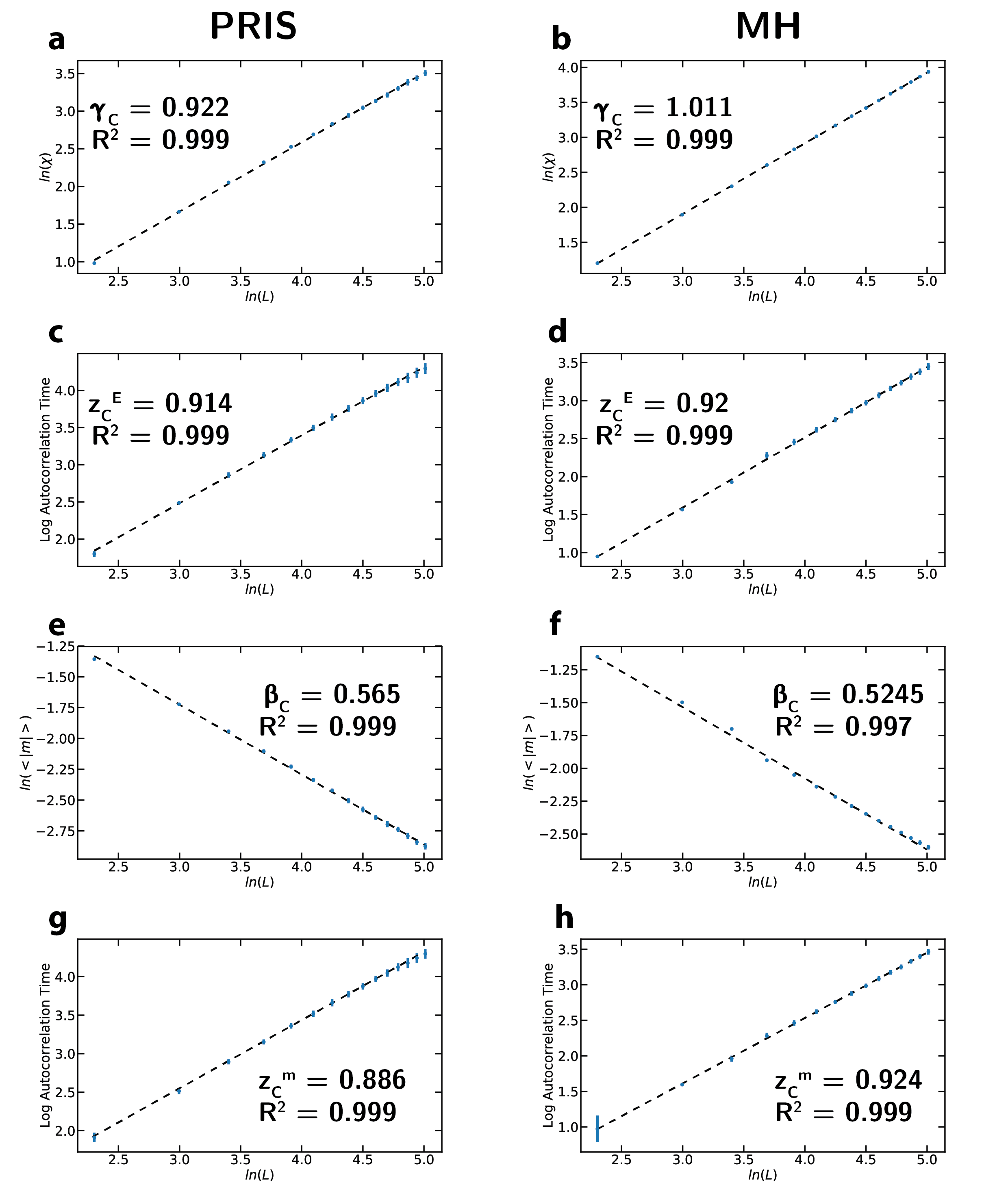}
\caption{\textbf{Probing the critical exponents of the infinite-range Ising model.} Fits are shown with the resulting critical exponent for the PRIS \textbf{(a, c, e, g)} and the MH \textbf{(b, d, f, h)} algorithms for the susceptibility \textbf{(a-b)}, energy autocorrelation time \textbf{(c-d)}, magnetization \textbf{(e-f)}, and magnetization autocorrelation time \textbf{(g-h)}.}
\label{fits_infinite_range}
\end{center}
\end{figure}

\begin{table}
\begin{center}
 \begin{tabular}{||c | c  c |c c |c c |c c||} 
 \hline
 Algorithm & $\beta_C$ & $R^2$ & $\gamma_C$ & $R^2$ & $z^E_\text{C}$ & $R^2$ & $z^m_\text{C}$ & $R^2$ \\ [0.5ex] 
 \hline\hline
       \multicolumn{9}{||c||}{\textbf{2D ferromagnetic Ising model}} \tabularnewline
       \hline
  MH & 0.1257& 0.992& 1.735&0.998 &1.393 &0.993 &2.068 &0.998 \\
  \hline
  PRIS & 0.1194& 0.972& 1.867& 0.990& 1.860& 0.977& 2.023& 0.994\\ [0.5ex] 
  \hline\hline
       \multicolumn{9}{||c||}{\textbf{Infinite-range Ising model}} \tabularnewline
       \hline
  MH & 0.5245& 0.997& 1.011& 0.999& 0.920& 0.999& 0.924& 0.999\\
  \hline
  PRIS & 0.5650& 0.999& 0.922& 0.999& 0.914& 0.999&0.886 &0.999 \\ [0.5ex] 
 \hline
\end{tabular}
\end{center}
\caption{\textbf{Summary of critical exponents measured with the PRIS and MH.} $R^2$ is the coefficient of determination of each power law fitting.}
 \label{criticalexp_table}
\end{table}

\newpage
\section{Supplementary Note 5: Scaling of the PRIS performance on several photonic architectures}
\subsection{Cascaded arrays of programmable Mach-Zehnder interferometers (MZI)}

\begin{figure}[h]
\centering
\includegraphics[width = 0.8\columnwidth]{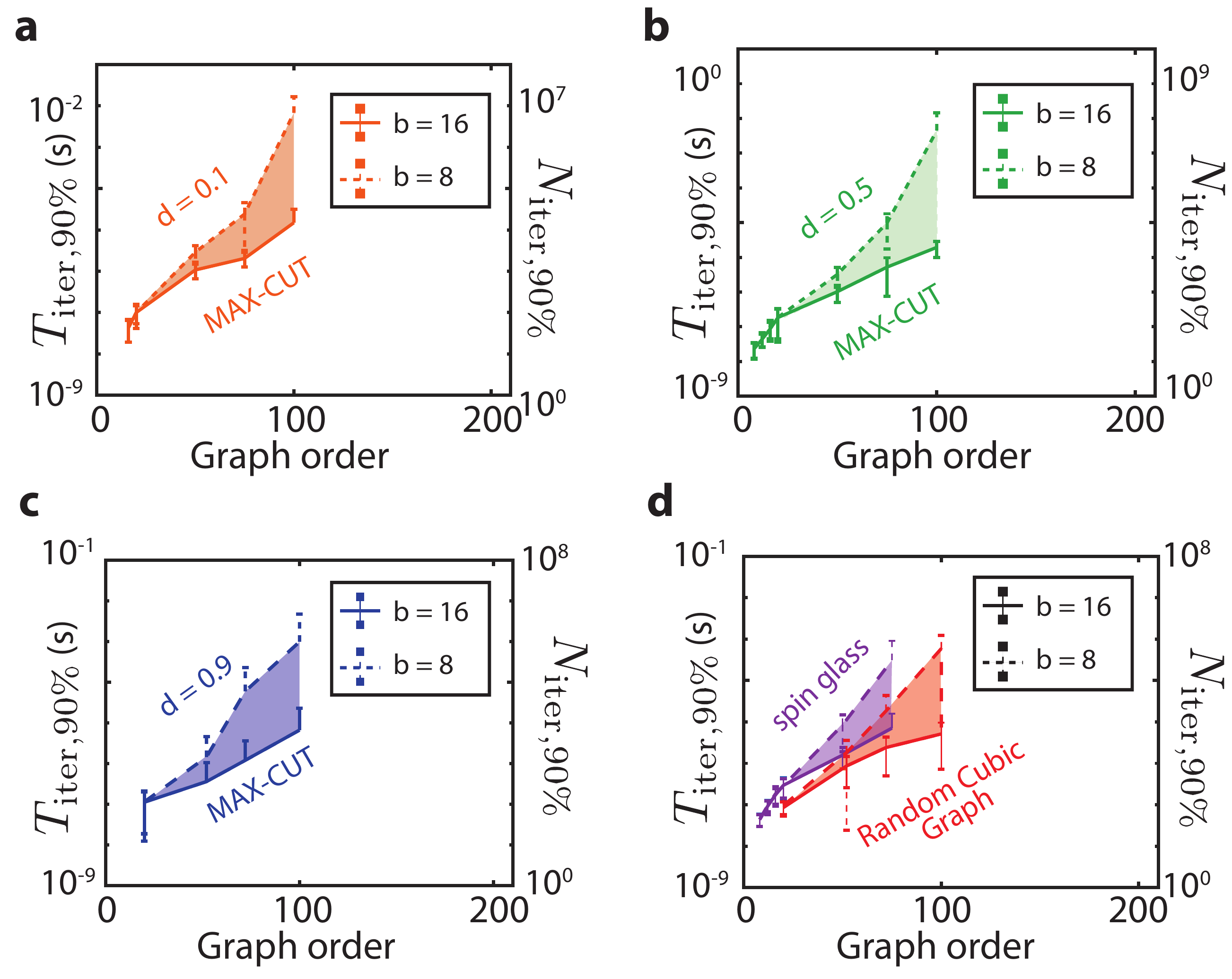}
\caption{\textbf{Simulated scaling of PRIS on a cascaded array of programmable MZIs.} Number of algorithm steps to reach a ground state with a probability of 90\% is plotted as a function of the graph order for various graph topologies: MAX-CUT graphs with densities $d = 0.1$ \textbf{(a)}, $d = 0.5$ \textbf{(b)}, and $d = 0.9$ \textbf{(c)}, Random Cubic Graphs (RCG), and spin glasses \textbf{(d)}.}
\label{fig:scaling}
\end{figure}

In Figure \ref{fig:scaling}, we examine the influence of the bit accuracy of the setting of the phases on the performance of the machine. This can describe any system where the matrix is encoded by a cascaded array of programmable MZIs such as \cite{Shen2017a,Harris2017}. We assume the phase can be set with $b$-bit accuracy, which means that when setting phase $\theta_m$ (resp. $\phi_m$) on the PNP, we actually draw from a uniform distribution $\theta \in [\theta_m-\frac{2\pi}{(2\pi)^b};\theta_m+\frac{2\pi}{(2\pi)^b}]$ (resp. $\phi \in [\phi_m-\frac{2\pi}{(2\pi)^b};\phi_m+\frac{2\pi}{(2\pi)^b}]$). 

Another way to describe the bit accuracy of the PNP is through the bit-accuracy of the voltage setting \cite{Shen2017a}. The phase-voltage relation can be approximated by a quadratic dependence: $ \theta = 2\pi\left(\frac{V}{V_{2\pi}}\right)^2$ where $V_{2\pi}$ is the voltage setting to achieve $2\pi$ modulation \cite{Harris2014}. If we assume the voltage is set with $b_V$ accuracy, then the actual voltage is drawn from a (uniform) distribution over $[V-\frac{V_{2\pi}}{2^{b_V}};V+\frac{V_{2\pi}}{2^{b_V}}]$. This translates to the phase accuracy as:
\begin{eqnarray}
\theta \pm \Delta \theta &=& 2\pi\left(V\pm \frac{V_{2\pi}}{2^{b_V}}\right)^2/V_{2\pi}^2\\
&=& \theta \pm \frac{\sqrt{2\pi\theta}}{2^{b_V-1}} + \frac{2\pi}{4^{b_V}}
\end{eqnarray}
We can safely neglect the last term and we notice that the worst case scenario corresponds to $\theta \sim 2\pi$, for which $b_V -1 = b_\theta$. A rule of thumb results: a $b_V$-bit accuracy of the PNP on its voltage setting corresponds to a $(b_V-1)$-bit accuracy of the PNP on its phase setting. Inversely, a $b_\theta$-bit accuracy of the PNP on its phase setting corresponds to a $(b_\theta+1)$-bit accuracy of the PNP on its voltage setting.

Static sources of noise could be a significant bottleneck in scaling the PRIS to large $N \sim 100$ graph orders. For instance, a static noise on the phase setting of an array of MZI will result in a static error on the effective coupling between spins, thus reshaping the Hamiltonian landscapes, which could impact the algorithm efficiency. We simulate the algorithm performance as a function of the graph order $N$ for phase resolutions of $b_\theta = 8$ and $16$ bits. The resulting time on a GHz photonic architecture to find the ground state with $90 \%$ chance, $T_{\text{iter}, 90\%}$, is also shown in Figure \ref{fig:scaling}. While a 16-bit phase resolution does not impact the algorithm performance (with scaling results comparable to an ideal photonic network, see main text), an 8-bit phase resolution may increase the required number of algorithm steps by one to two orders of magnitude, depending on the graph order and topology (while still outperforming other photonic systems on a GHz architecture, such as \cite{McMahon2016, Hamerly2018}). Thus, the reduction of static noise is of paramount importance in the realization of the PRIS on large-scale static photonic networks. 

\subsection{Optical Neural Networks based on Photoelectric Multiplication}
For larger graphs $N \sim 10^3 - 10^6$, one could resort to recently-proposed large-scale optical neural networks based on photoelectric multiplication \cite{Hamerly2018b}. By encoding both matrix weights $C_{ij}$ and input signals $S^{(t)}_i$ into optical (time) domain, the measured output is added a Gaussian noise term with amplitude defined by the number of photons per Multiply And Accumulate (MAC):
\begin{equation}
S^{(t+1)}_i = \text{Th}_\theta \left( \sum_j C_{ij} S_j^{(t)} + w_i^{(t)} \frac{|| C || || S^{(t)}||}{N^{3/2}}  \frac{\sqrt{N}}{\sqrt{n_\text{mac}}}\right).
\end{equation}
For the various problems we study in this paper, the working standard deviation is usually $\sim 1$. This corresponds to a working $n_\text{mac}$ of
\begin{equation}
n_\text{mac} \sim N \frac{||C||^2 ||S^{(t)}||^2}{N^3}
\end{equation}
We can evaluate the corresponding total energy consumption per matrix multiplication for 10 random spin glasses. We get $n_\text{mac} \sim 4$ (resp. $\sim 15$) for $N = 100$ (resp. $N =1,000$). Smaller working $n_\text{mac}$ will be required for sparser graphs, since $||C||$ is smaller. The corresponding SNR scales as $\sim n_\text{mac}$ and the total energy is $N^2 n_\text{mac} \sim 6.2 \pm 0.35$ fJ/matrix multiplication (resp. $1.9 \pm 0.5$ pJ/matrix multiplication).

There are many attractive features of these networks for the implementation of large-scale PRIS:
\begin{enumerate}
\item[$\triangleright$] This architecture naturally leverages quantum noise which perturbs the output as is required for the good execution of the PRIS. The noise level can be tuned by changing the number of photons per MAC which is proportional to the SNR. 
\item[$\triangleright$] The non-linear function is executed in electronic domain, which allows a lot of flexibility on its implementation and reconfiguration (the threshold function required for the good operation of the PRIS would be straightforward to implement), while optical nonlinearities working at low-power have not been demonstrated so far. 
\item[$\triangleright$] This architecture is in principle scalable to very large number of spins $N \sim 10^6$.
\end{enumerate}

\subsection{Free space optical Neural Networks}

Since the PRIS relies on a static transformation, the use of free-space optical neural networks \cite{Farhat1985, Lin2018} with 3D printed masks or reconfigurable SLMs is another option to achieve PRIS with $N \sim 10^6$ neurons. The analysis we made on the influence of heterogeneities (see Figure \ref{fig:scaling}) remains relevant here. A set of lens - mask - lens would realize the matrix multiplication on the signal $S^{(t)}$ encoded in optical domains, while the coupling matrix $C$ is encoded in the mask transmission. The speed of such a free-space architecture would only be limited by the photodetector (typically $\sim 10$ THz) and modulation (typically $\sim 1$ kHz for SMLs, $\gtrapprox 1$ GHz ) bandwidths.

\subsection{Comparison table of various heuristic algorithms and architectures}

\begin{table}
\begin{center}
 \begin{tabular}{||c c c c c ||} 
 \hline
Algorithm & Algorithm step & Architecture & Algorithm step time estimate with $N = 100$ & Time Complexity \\
 \hline\hline
MH & Sequential & CPU & $\sim 1$ ms$^{*}$ & $O(N^2)$ \\
 \hline
\multirow{ 5}{*}{PRIS} & \multirow{ 4}{*}{Parallelizable} & CPU & $\sim 10-100 \mu$s$^{*}$ & $O(N^2)$ \\
 &  & FPGA & $\sim 5-100$ns$^{**}$ & $O(N^2/M)$ \\
&  & MZI network \cite{Shen2017a, Harris2017}  & $\sim 0.1 - 1$ns & $O(N)$\\
& & \multirow{ 2}{*}{Free space optics \cite{Hamerly2018, Pierangeli2019}} & modulation bandwidth limited & \multirow{2}{*}{$O(1)$} \\
& & & (SLM $\sim 1$ms, electro-optic modulators $\sim 0.1 - 1$ns) & \\
 \hline
 \end{tabular}
\end{center}
\caption{\textbf{Comparative table of projected performance of various heuristic algorithms implemented on various architectures.} Time complexity is given as a function of the total number of spins N. Regarding clock estimates: $^{*}$ Estimate run on a 2.7 GHz Intel Core i5 with Matlab, $^{**}$ Estimate run on a Xilinx Zynq UltraScale+ MPSoC ZCU104 with a systolic array architecture. See Supplementary Note 7 for more details.}
 \label{perf_comp}
\end{table}

Since the PRIS essentially relies on fast vector-to-fixed matrix multiplication, it can be implemented efficiently on various photonic and parallel electronic hardware architectures, such as FPGAs. We summarize in Table \ref{perf_comp} the projected performance of various heuristic algorithms (MH and PRIS) running on several hardwares (both electronics and photonics). We observe that, while photonics potentially allows the fastest clock for such algorithms and a competitive scaling factor, FPGAs can achieve similar clocks and require much less engineering (they can be bought off the shelf and are easily reconfigurable). To demonstrate our point, we perform a proof-of-concept experiment on a FPGA board, whose results are shown in Supplementary Note 7.

If photonic architectures can significantly reduce the time complexity of the algorithm step by performing massive multiplexing\cite{Pierangeli2019}, it must be noted that we are neglecting time overhead such as fabrication, etc. Also, some free-space architectures, such as SLM, are typically slow to reconfigure (on the order of 1ms), thus only being relevant for multiplying very large matrices.  

\color{black}

\newpage
\section{Supplementary Note 6: Comparison of the PRIS to several (meta)heuristics}
\subsection{Benchmarking versus Metropolis-Hastings on large spin glasses $N\sim 1000$}

\begin{figure}
\begin{center}
\includegraphics[scale=0.4]{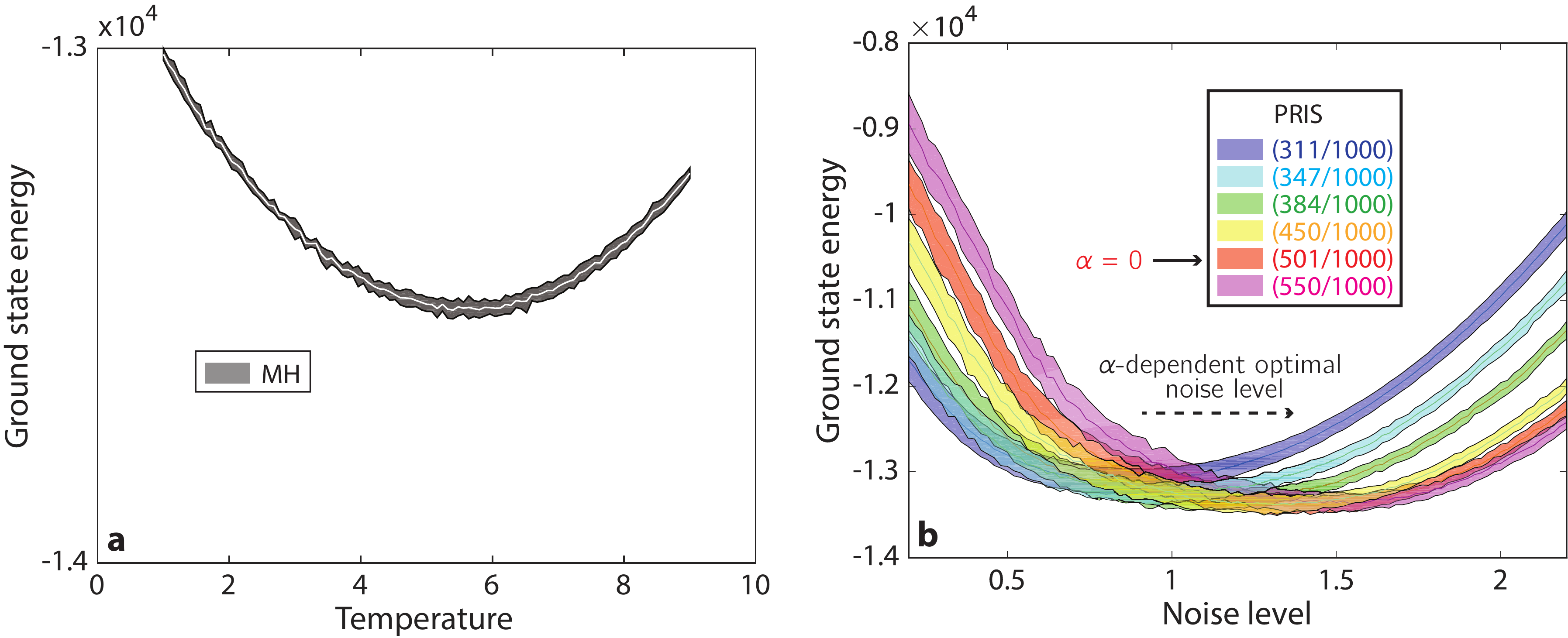}
\caption{\textbf{Benchmarking the PRIS versus MH on large spin glasses $N = 1000$.} For both PRIS and MH, the algorithm is ran 100 times for 10,000 iterations at each temperature / noise level. The line shows the average ground state energy over 100 runs, and the shaded area corresponds to $\pm$ the standard deviation. PRIS is ran for various dropout levels corresponding to $(N_\text{eig}/N)$ eigenvalues (see legend). Results for MH and PRIS at the smallest eigenvalue dropout level (311/1000) are also averaged over the 10 random spin glasses. Results for the PRIS at all others eigenvalue dropout levels are only averaged over spin glass 1 (see discussion). For this study, we choose $\Delta_{ii} = |\sum_j K_{ij}|$.}
\label{sg_N1000}
\end{center}
\end{figure}

In order to evaluate the performance of the PRIS on larger graphs, for which exact solvers typically fail, we benchmark the PRIS against MH for a set of 10 random spin glasses (whose couplings are randomly chosen from a uniform distribution in the interval $[-1, 1]$). First, we notice that the behavior probed by the PRIS and MH is very similar over the 10 spin glasses. The temperature/noise-dependent mean ground state energy found by MH and the PRIS is shown in Figure \ref{sg_N1000}. In particular, we see that the PRIS achieves mean ground state energies similar to MH. As a reminder, a heuristic mapping between the temperature from statistical physics and the effective temperature (noise level) of the PRIS is $T \sim (k\phi)^2$. These graphs orders are large enough, so that standard exact solvers are taking a long time to find the ground state. For instance, we ran BiqMac on spin glass 1 on their online submission server\cite{Rendl2010} for three hours. The algorithm could only find an approximate ground state which was outperformed by MH, PRIS-A, and SA. 

We also observe that the PRIS shows an $\alpha$-dependent optimal noise level, which increases with $\alpha$. For spin glass 1, the absolute lowest ground state energy is obtained for $\alpha = -0.2$, i.e. (450/1000) eigenvalues. This hints at the necessity of optimizing the hyperparameter $\alpha$ around $\alpha \sim 0$ when running the PRIS on large graphs. 

\subsection{PRIS-A: a proposed metaheuristic variation of the PRIS}

\begin{table}
\begin{center}
 \begin{tabular}{||c c c c c c||} 
 \hline
Instance & PRIS & MH & PRIS-A & SA  & BiqMac \\
 \hline\hline
spin glass 1 & -13,716 & -13,796 & -13,795 & -13,814 & -13,746 \\
 \hline
spin glass 2 & -13,632 & -13,736 & -13,716 & -13,754 & N.C.\\
 \hline
spin glass 3 & -13,717 & -13,777 &-13,760 & -13,796 & N.C.\\
 \hline
spin glass 4 & -13,678 & -13,768 &-13,781 & -13,798& N.C.\\
 \hline
spin glass 5 & -13,732 & -13,755 &-13,767 & -13,782 & N.C.\\
 \hline
 spin glass 6 & -13, 752 & -13,828 &-13,804 & -13,832 & N.C.\\
 \hline
 spin glass 7 & -13,723 & -13,792 &-13,758 & -13,800 & N.C.\\
 \hline
 spin glass 8 & -13,731 & -13,769 &-13,781 & -13,783 & N.C.\\
 \hline
 spin glass 9 & -13,711 & -13,810 &-13,798 & -13,817 & N.C.\\
 \hline
 spin glass 10 & -13,754 & -13,846 & -13,822 & -13,855 & N.C.\\
 \hline
 \end{tabular}
\end{center}
\caption{\textbf{Summary of benchmarking PRIS and PRIS-A against MH and SA.} For both PRIS and MH, the algorithm is ran 100 times for 10,000 iterations at each temperature / noise level. The table shows the absolute lowest ground state energy recorded. For PRIS-A and SA, the algorithm is ran 100 times with $N_\text{alg, iter}$ temperature increments (as given by Eq.\eqref{eq:N_alg_iter}) and $N_\text{iter per temp.}=100$. For PRIS-A (resp. SA), the initial noise level (resp. temperature) is $\phi_i = 50$ (resp. $T_i = 5,000$), the final noise level is $\phi_f = 0.1$ (resp. $T_f = 0.01$) and the temperature geometric factor is $\lambda = \sqrt{0.991}$ (resp. $\lambda = 0.991$). For PRIS-A, the eigenvalue dropout level is taken to be $\alpha = 0$, corresponding to (501/1000) eigenvalues. For BiqMac, the algorithm ran for three hours (time limit on the BiqMac online job submission platform). N.C. = Non Computed.}
 \label{pris_perf}
\end{table}

A route to achieving systematically low energy states is to design metaheuristics, i.e. master strategies guiding the search for an optimal ground state. Simulated Annealing (SA) \cite{Kirkpatrick1983} is one of such algorithms, derived from Metropolis-Hastings. Let us reminder the reader of one possible implementation of this algorithm with the widely used geometric schedule for the temperature \cite{Kirkpatrick1983, Nourani1998, Glover1998, Cohn1999}: 

\begin{algorithm}
\begin{algorithmic}
\STATE Start from random initial state
\STATE Choose initial temperature $T_i$ and geometric factor $\lambda <1$
\STATE $T \leftarrow T_i$
\STATE \textbf{for all} $i \in \{1, ..., N_\text{alg, iter} \}$
\STATE \hskip1em $T \leftarrow \lambda T$
\STATE \hskip1em \textbf{for all} $j \in \{1, ..., N_\text{iter per temp.} \}$
\STATE \hskip2em Update state according to MH acceptance rule at temperature $T$
\STATE \hskip1em \textbf{end for}
\STATE \textbf{end for} 
\end{algorithmic}
\caption{Simulated annealing algorithm with a geometric schedule.}
\label{sa_algo}
\end{algorithm}

This naturally inspires a metaheuristic based on the PRIS, which we call Photonic Recurrent Ising Simulated Annealing (PRIS-A):
\begin{algorithm}
\begin{algorithmic}
\STATE Start from random initial state
\STATE Choose initial noise level $\phi_i$ and geometric factor $\lambda <1$
\STATE $\phi \leftarrow \phi_i$
\STATE \textbf{for all} $i \in \{1, ..., N_\text{alg, iter} \}$
\STATE \hskip1em $\phi \leftarrow \lambda \phi$
\STATE \hskip1em \textbf{for all} $j \in \{1, ..., N_\text{iter per temp.} \}$
\STATE \hskip2em Update state according to PRIS acceptance rule at noise level $\phi$
\STATE \hskip1em \textbf{end for}
\STATE \textbf{end for} 
\end{algorithmic}
\caption{Photonic Recurrent Ising Simulated Annealing (PRIS-A) algorithm.}
\label{sa_algo}
\end{algorithm}

Let us note a couple of peculiarities: 
\begin{enumerate}
\item[$\triangleright$] The noise level from the PRIS is related to an effective temperature via $T \sim \phi^2$. Thus, one should compare SA ran with a geometric factor $\lambda$ to PRIS-A with a geometric factor $\sqrt{\lambda}$.
\item[$\triangleright$] In the PRIS-A algorithm, the eigenvalue dropout level $\alpha$ is also a degree of freedom. One could thus, in principle, also simulate the annealing of the eigenvalue dropout level, thus affecting the ground state search dimensionality. A comprehensive study of this new class of algorithms goes beyond the scope of this work. 
\item[$\triangleright$] For given initial and final noise levels/temperatures and geometric factors, one can determine the number of temperature increments $N_\text{alg, iter}$ with the formula:
\begin{equation}
\label{eq:N_alg_iter}
N_\text{alg, iter} = \frac{\log \phi_f - \log \phi_i }{\log \lambda}.
\end{equation}
\item[$\triangleright$] $\lambda$ is typically chosen to be smaller but close to 1, in order to mimic adiabatic temperature variations. We verify that both for SA and PRIS-A, $\lambda > 0.98$ yields consistently low energy ground states, with no particular amelioration when increasing $\lambda$ (and scaling the number of temperature increments $N_\text{alg, iter}$).
\end{enumerate}

The initial and final noise levels chosen for PRIS-A are $\phi_i = 50$ and $\phi_f = 0.1$. The initial and final temperatures chosen for SA are $T_i = (2 * 0.5877 * \phi_i)^2 \sim 3454$ and $T_f =  (2 * 0.5877 * \phi_f)^2 \sim 0.01$. We observed no significant variation on the minimum ground state energy found for $\lambda > 0.99$ and ran each algorithms with a rate of $\lambda = 0.991$. The performance of the various algorithms we implement is shown in Table~\ref{pris_perf}. MH yields lower ground state energies than PRIS on average of 0.53\%. However, PRIS-A can outperform PRIS by a similar amount, lowering energies on average of 0.46\%. This is a larger performance enhancement than SA to MH (0.11\% decrease of ground state energy). Then, on average, SA outperforms PRIS-A by 0.18\%. We expect optimization of the eigenvalue dropout level $\alpha$ (and simulated annealing on this parameter) to further enhance the performance of PRIS and PRIS-A.

\newpage

\section{Supplementary Note 7: Implementation of the PRIS on FPGA}

\subsection{Architecture}

The PRIS algorithm is primarily designed for future photonic implementations. However, both photonic chips and FPGAs (Field Programmable Gate Arrays) share parallel processing capability. Thus, it is worthwhile to first demonstrate the performance of the algorithm on an FPGA board.

We have implemented the algorithm on Xilinx Zynq UltraScale+ multiprocessor system-on-chip (MPSoC) ZCU104 evaluation board based on the Pynq framework. The high-level synthesis, place, and route have been performed by Xilinx Vivado 2018.3 design suite. Zynq and Zynq Ultrascale+ devices integrate a multi-core processor (ARM Cortex-A9) and programmable logic (FPGA) in the same circuitry. PYNQ is an open-source project released by Xilinx that enables Python Productivity for Zynq devices. It incorporates the open-source Jupyter notebook infrastructure to run an Interactive Python (IPython) kernel and a web server directly on the ARM processor of the Zynq device. It also provides extensive hardware libraries (overlays) and APIs which enable easier and faster programming of FPGA. 

For our application, the preprocessing and preparation of data is performed in a Jupyter notebook using Python. Once ready, we transmit the data and write to memories on FPGA through AXI interface. After running the recurrent algorithm for a number of cycles on FPGA, the final state vector is transmitted back for data analysis like energy calculation and correctness verification. In this manner, the same hardware configuration could be easily reconfigured and run different problems efficiently. Noise generation and post-processing (energy calculation, and more generally extracting observable) could also be run directly on the FPGA in future versions of this implementation. 

The overall high level architecture is shown in Figure \ref{architecture}. We make use of the address-mapped AXI interface to encode different operations (e.g. Block RAM (BRAM) addressing, loop setting, result selection) into different addresses. AXI controller loads matrix, noise and threshold data into BRAMs, then sets up the initial state vector and starts the computation. The loop module is a finite-state-machine which reads the data from BRAM, adds noise, compares with threshold and then updates the state vector at every loop step. The final state vector and clock information is transmitted back through AXI. 

\begin{figure}
\begin{center}
\includegraphics[scale=0.4]{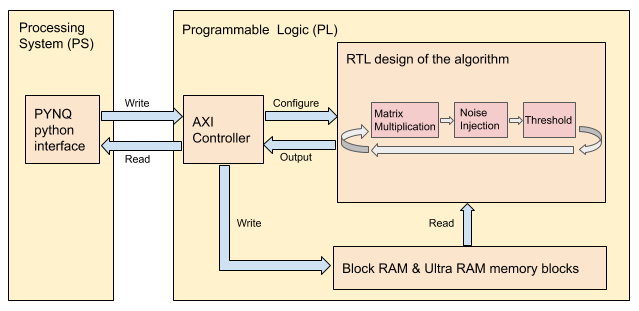}
\caption{\textbf{High level architecture of the design.}}
\label{architecture}
\end{center}
\end{figure}

We emphasize the following considerations regarding our implementation: 
\begin{enumerate}
\item[$\triangleright$] The bottleneck of the performance implementation is to read input matrix data from memory. The ZCU104 board has 312 BRAM blocks, with each block transmitting a maximum of 72 bits per clock cycle when configured under true-dual-port mode. It also has 94 UltraRAM (URAM) blocks with each block transmitting a maximum of 144 bits per clock cycle. Thus the total maximum data rate it could read from BRAM and URAM is 36288 bits per clock cycle. For our current implementation ($N<100$) we could assume we have enough memory, but as the problem size increases this becomes a major limitation.
\item[$\triangleright$] We multiply an $N$-by-$N$ $b$ bit matrix ($b$ being the bits we choose to represent the Ising problem data) with an $N$-by-$1$ $1$-bit state vector, so each data in the result is a conditional sum of one row of matrix $C$ based on the value of the state vector. We choose to implement a binary tree for speeding up the sum as shown in Figure \ref{bintree}. At each time step, a certain number of rows of the matrix is loaded onto the leaves of the tree, and then gets propagated to next level based on the value of the state vector, and adds all up thereafter. The result of the sum should propagate to the root after a delay of clock cycles equal to the height of the tree. 
\end{enumerate}

\begin{figure}
\begin{center}
\includegraphics[scale=0.65]{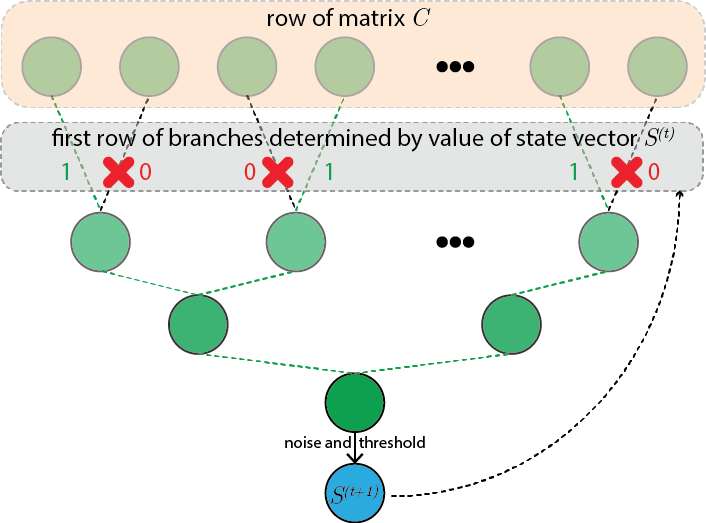}
\caption{\textbf{Binary tree architecture for matrix-to-vector multiplication.} A row of matrix C is read from memory to the top leaves of the binary tree. The branch connecting the top leaves to the next tree layer are either active or not, depending on the value of the current state vector (0 or 1).}
\label{bintree}
\end{center}
\end{figure}

The loop is carefully pipelined, in which reading the matrix takes $\frac{N}{2b}$ clock cycles, multiplication (binary tree addition) takes $\log_2 N$ cycles, while noise injection and thresholding are done in a single clock cycle. Details of the complexity is discussed in the next session.
We are only running integer arithmetic on the FPGA, so to adapt the original algorithm which runs on float points, we need to scale the parameters (matrix, noise, and threshold) and then round them to integers. The scaling factor and bit length to represent each data needs to be carefully designed.

Note that the only difference between our implemented algorithm on the FPGA and the ideal algorithm is that we round our values to a specific bit depth. Specifically, we store the matrix values, noise, and threshold values only to a specific bit precision. This originates from memory constraints on the FPGA, but it is reasonable to suspect that this rounding affects the performance of the algorithm. To test the effects of this change, we implemented both our original algorithm and the version with rounding on MATLAB, testing various bit depth with varying input, as shown in Figure \ref{comparison}. We found that, for a matrix of size $N = 100$ by 100, rounding after multiplying by 64 or even 32 performed just as well as using arbitrary bit precision (Matlab's default double-precision floating point). Specifically, for a variety of matrix inputs, on average using 64 and 32 roundings reached states of essentially equivalent energy to the arbitrary precision. Over 1000 trials each with a different input matrix, 1000 iterations each (i.e. 1000 matrix multiplications and noise additions and thresholding for each trial), on average our normal algorithm reached a minimum energy of -408.63, whereas the 32 rounding reached -409.08, and the 32 bit reached -409.17. 

\begin{figure}
\begin{center}
\includegraphics[scale=0.4]{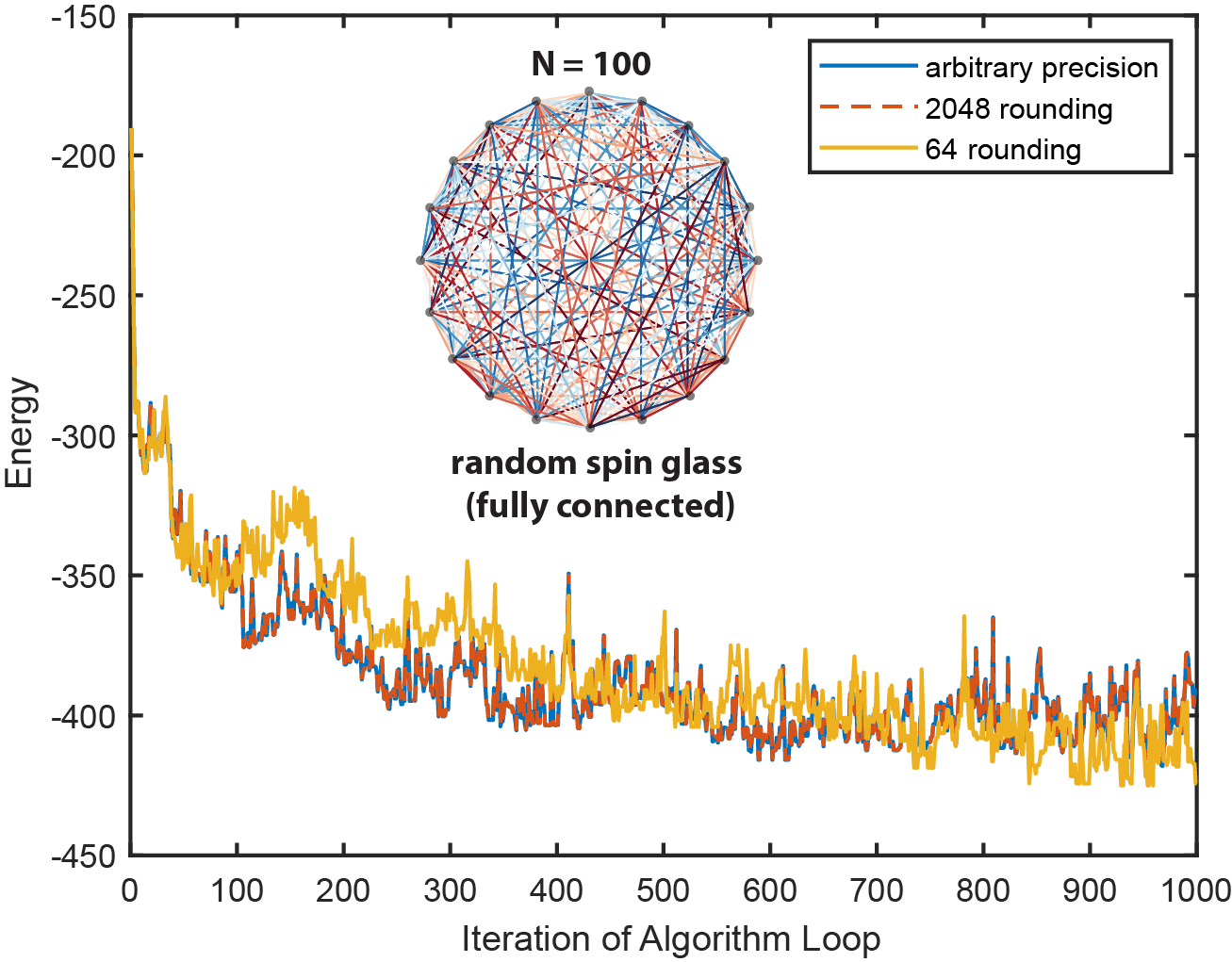}
\caption{\textbf{Performance comparison for various rounding values.}}
\label{comparison}
\end{center}
\end{figure}

\subsection{Complexity}

In discussing the complexity of our implementation, we must note the dependence on the particular resources of a given FPGA. Here we will discuss the complexity of a single step of our algorithm, consisting of a single matrix multiplication followed by noise addition and threshold comparison. Since the clock speed of any FPGA is variable, we give the complexity in terms of clock cycles.

We focus on computing the complexity of the matrix multiplication, since it dominates the three steps. To perform this matrix multiplication, we must read the entire matrix from memory and feed it to a system of binary trees. Suppose we store each matrix entry using $b$ bits, and that we are working with an $N \times N$ size matrix. In our FPGA, we can read 64 bits from each BRAM block per clock cycle, so one limiting factor is the total number $k$ of BRAM blocks we can use concurrently. Note that we are also limited by the maximum number $R$ of bits that can be held in registers, or the working memory, of our FPGA. This allows us to compute a complexity bound.

Assuming the FPGA can read $M$ bits from each BRAM block per clock cycle, it will take the FPGA at least $\frac{N^2b}{M}$ cycles to read the entire matrix. According to the above considerations, $M =64k + 144u$, where $k$ (resp. $u$) is the number of used BRAM (resp. URAM). In our current implementation, we only used BRAM memory so $u=0$. Alternatively, we are also limited by the amount of data we can work with at a time, $R$. To complete this matrix multiplication, we must store a binary tree corresponding to each row. The result of these binary trees is a vector of $N$ values to be sent to the noise addition step. Thus the total amount of data we will need to work with in this computation is $2N^2b$ bits, which will take at least $\frac{2N^2b}{R}$ cycles to pass through the registers of the FPGA.

To finish the analysis, note that we may feed new data into the binary trees before the original data has fully propagated through. Thus the only computation time outside of that required to load the matrix is the time for the final data to propagate through the binary trees. The size of our binary trees can be determined by considering the minimum of the amount of data we can load at a time, $M = 64k$, and the amount of data (in bits) we can work with at a time, $R$. Note that we never need a binary tree larger than the size of a row of the matrix, and that the binary tree takes up twice the space of the data inputted into it. Then the size of our binary tree is $B:=\min(\frac{64k}{b},\frac{R}{2b},N)$, and it takes on the order of $\log_2 B$ cycles to propagate through. Then, incorporating the time to load the matrix, the total number of cycles for a single algorithm step is on the order of 
\[ \text{number of clock cycles per PRIS algorithm iteration} = \max \left( \frac{2N^2b}{R},\frac{N^2b}{64k} \right) + \log_2 B.\]

Lastly, considering noise addition and thresholding, note that this only involves reading $2Nb$ more bits from bram. The actions of noise addition and thresholding themselves only take two clock cycles, and the time and space necessary to deal with these extra $2Nb$ bit is negligible in complexity compared to that necessary to deal with the $N^2b$ bits from the matrix. So these two steps do not change the overall complexity. We also need to load values into the BRAM in the first place, but we do not include this process in the analysis. Thus the overall complexity of a single iteration of our algorithm is $O\left(\max(\frac{2N^2b}{R},\frac{N^2b}{64k}) + \log_2 B \right)$.

The discussion above is under the assumption that for large size problems ($N>100$), board registers are not enough for holding all the data in the matrix. For small size problem, we could alternatively use registers and do the multiplication in one clock cycle, then the major time consumption would be the binary tree addition, which took $\log_2 B$ clock cycles. 

For problems with larger size, we store matrix in BRAM as mentioned above. We could plot the complexity (Figure \ref{complexity}) discussed above under different assumptions: (1) Using all of the memory blocks on the FPGA (Figures Figure \ref{complexity}(a) and (b)), which is optimal but takes time to implement for every problem size. (2) For most reasonably-sized problems, we can assume a linear scaling of memory blocks, i.e. $M = r N b$, which is sub-optimal for a given $N$ but easier to reconfigure for various $N$ allowing this approximation (we can switch between different problem sizes $N$ by changing the width of BRAM IPs in our design). We also assumed that $R$ is large, thus not being the limiting factor of the computation. 

\begin{figure}
\begin{center}
\includegraphics[scale=0.4]{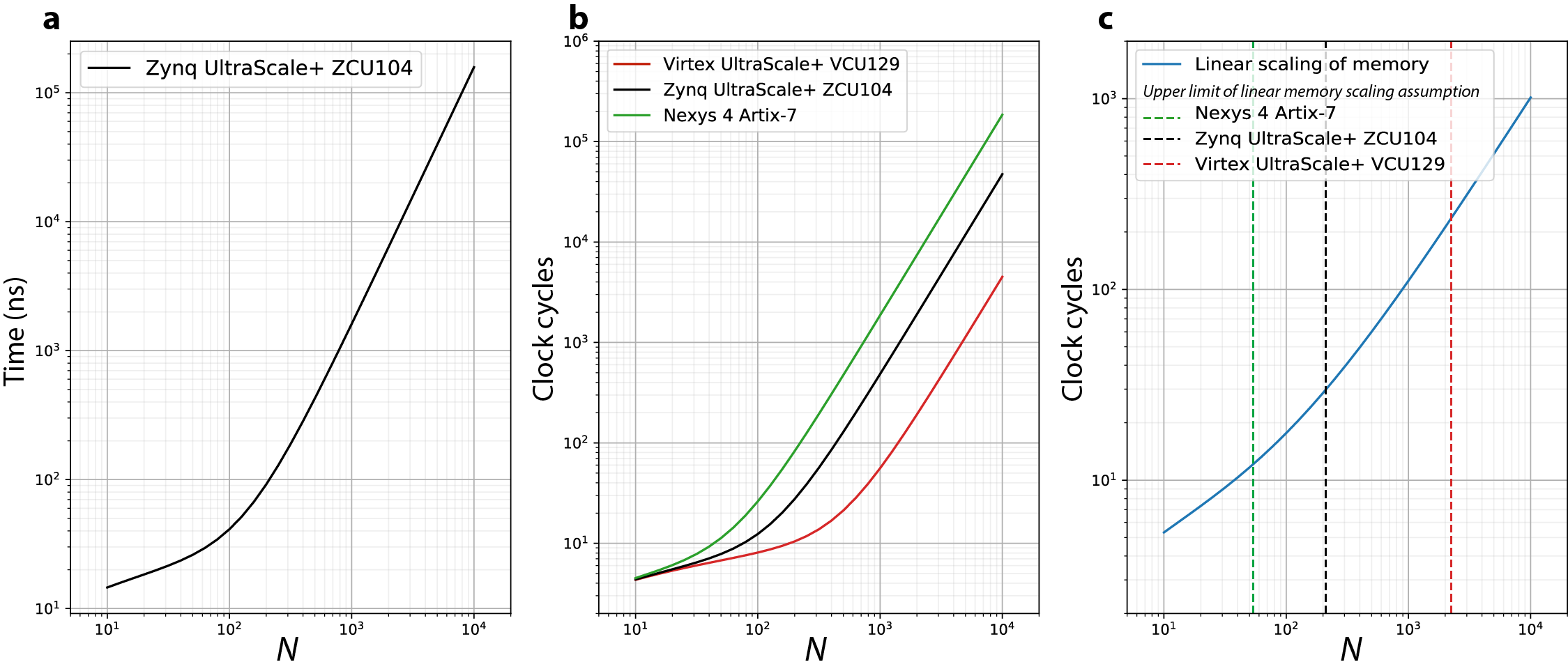}
\caption{\textbf{Complexity estimate of a PRIS single algorithm iteration, implemented on FPGA. a:} Graph showing the estimated time consumption on the current ZCU104 board, given that we’re using all of the memory blocks and under a clock of 300MHz and 16 bit encoding. \textbf{b:} Comparison of three different FPGA boards with different memory resources. \textbf{c:} Assume linear scaling of memory, i.e. $M = r N b$, here we let $r = 10$ (read 10 rows of matrix at a time). The vertical lines indicate the limit of each FPGA. For our ZCU104 board, the maximum problem size under the linear scaling assumption is $N=212$.}
\label{complexity}
\end{center}
\end{figure}

\subsection{Experiment}

In this experiment, we tested the correctness and runtime for a problem size of $N = 100$ under a 300 MHz clock. We used 16 bits to encode each data in matrix, noise and threshold. A total of $k=5$ BRAM IPs are instantiated in the design, each with data length of 1600 bits, corresponding to one row of the matrix. (In this design we are using 71\% of BRAM blocks on our board.)

The experiment runs as follows:
\begin{enumerate}
\item[$\triangleright$] Original Test: First we generate a test case running on float numbers in Jupyter notebook using Python. We run the test for 200 clock cycles and record the state vector and energy of each iteration. 
\item[$\triangleright$] Scale the problem: We scale up the matrix, noise and threshold by a factor of 128 and round to integers. We run the scaled test for 200 clock cycles and plot the energy array together with the original test to check the accuracy of the algorithm with rounded values. Several scaling factor and energy arrays are shown in Figure \ref{comparison}, in which we could see that scaling under 64 shows some deviation to the unrounded algorithm, and scaling above 1024 results in exactly the same sequence of state vectors.
\item[$\triangleright$] Correctness verification: we load the scaled test onto FPGA, and run 200 clock cycles and output state vector at each time step, with which we calculate the energy in Python. The FPGA-simulated energy and Python calculated energy is exactly the same, confirming the correctness of the implementation. 
\item[$\triangleright$] Timing and complexity analysis: for this part we only output the state vector and clock count at the very end of 200 clock cycles. For our problem of size 100 by 100, a loop of 200 iterations cost 3803 cycles. On average each time step takes 19 clock cycles (the remaining 3 clock cycles are due to some overhead at the beginning of the run), which is in accordance with our analysis before. 
\end{enumerate}

\begin{figure}
\begin{center}
\includegraphics[scale=0.3]{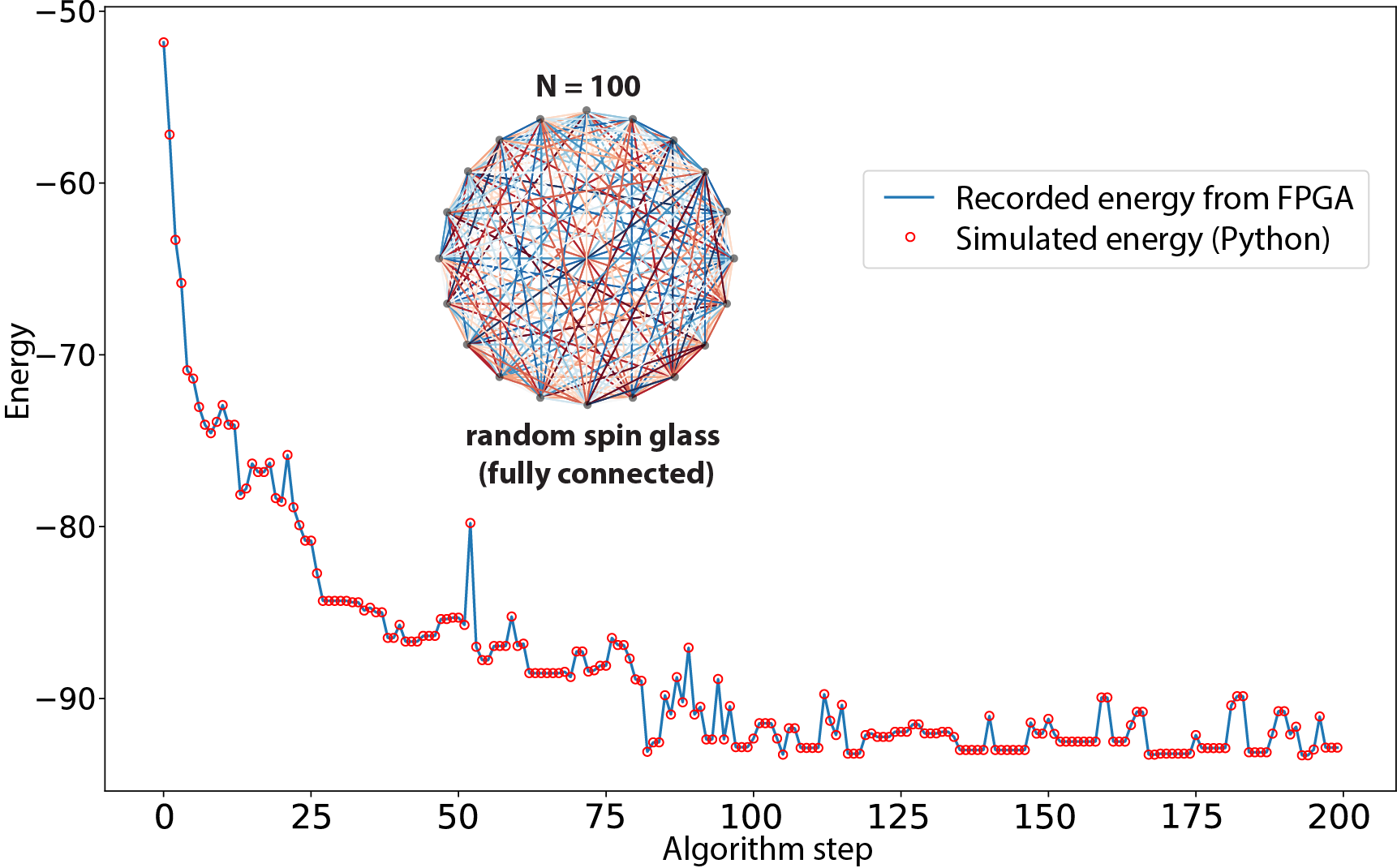}
\caption{\textbf{Comparison of FPGA and simulated (Python) outputs for a given random spin graph with $N=100$.}}
\end{center}
\end{figure}

The measured clock cycle per algorithm step is 19 cycles, which is exactly as expected: 19 = 10 + 8 + 1 in which 10 is the cycles it takes to read the whole matrix (we are reading $r=10$ rows at each clock cycle), 8 is the time for the binary addition tree, and an additional clock is required for noise injection and thresholding. In conclusion, we have implemented the PRIS for problem size $N=100$ on an FPGA platform with a time per algorithm step of approximately 63 ns. 

\color{black}

\newpage 
\bibliographystyle{ieeetr}
\bibliography{../ising}

\begin{thebibliography}{10}

\bibitem{Landau2009}
D.~P. Landau and K.~Binder, {\em {A Guide to Monte Carlo Simulations in
  Statistical Physics}}.
\newblock 2009.

\bibitem{Hromkovic2013}
J.~Hromkovi{\^{c}}, {\em {Algorithmics for Hard Problems: Introduction to
  Combinatorial Optimization, Randomization, Approximation, and Heuristics}}.
\newblock Springer Berlin Heidelberg, 2013.

\bibitem{Kardar1986}
M.~Kardar, G.~Parisi, and Y.-C. Zhang, ``{Dynamic Scaling of Growing
  Interfaces},'' {\em Physical Review Letters}, vol.~56, pp.~889--892, mar
  1986.

\bibitem{Isichenko1992}
M.~B. Isichenko, ``{Percolation, statistical topography, and transport in
  random media},'' {\em Reviews of Modern Physics}, vol.~64, pp.~961--1043, oct
  1992.

\bibitem{Honerkamp-Smith2009}
A.~R. Honerkamp-Smith, S.~L. Veatch, and S.~L. Keller, ``{An introduction to
  critical points for biophysicists; observations of compositional
  heterogeneity in lipid membranes},'' {\em Biochimica et Biophysica Acta (BBA)
  - Biomembranes}, vol.~1788, pp.~53--63, jan 2009.

\bibitem{Albert2002}
R.~Albert and A.-L. Barab{\'{a}}si, ``{Statistical mechanics of complex
  networks},'' {\em Reviews of Modern Physics}, vol.~74, pp.~47--97, jan 2002.

\bibitem{Glover2006}
F.~Glover and G.~Kochenberger, {\em {Handbook of Metaheuristics}}.
\newblock Springer, 2006.

\bibitem{Wang2013}
Z.~Wang, A.~Marandi, K.~Wen, R.~L. Byer, and Y.~Yamamoto, ``{Coherent Ising
  machine based on degenerate optical parametric oscillators},'' {\em Physical
  Review A}, vol.~88, p.~063853, dec 2013.

\bibitem{McMahon2016}
P.~L. McMahon, A.~Marandi, Y.~Haribara, R.~Hamerly, C.~Langrock, S.~Tamate,
  T.~Inagaki, H.~Takesue, S.~Utsunomiya, K.~Aihara, R.~L. Byer, M.~M. Fejer,
  H.~Mabuchi, and Y.~Yamamoto, ``{A fully programmable 100-spin coherent Ising
  machine with all-to-all connections.},'' {\em Science (New York, N.Y.)},
  vol.~354, pp.~614--617, nov 2016.

\bibitem{Wu2014}
K.~Wu, J.~{Garc{\'{i}}a de Abajo}, C.~Soci, P.~{Ping Shum}, and N.~I. Zheludev,
  ``{An optical fiber network oracle for NP-complete problems},'' {\em Light:
  Science {\&} Applications}, vol.~3, pp.~e147--e147, feb 2014.

\bibitem{Vazquez2018}
M.~R. V{\'{a}}zquez, V.~Bharadwaj, B.~Sotillo, S.-Z.~A. Lo, R.~Ramponi, N.~I.
  Zheludev, G.~Lanzani, S.~M. Eaton, and C.~Soci, ``{Optical NP problem solver
  on laser-written waveguide platform},'' {\em Optics Express}, vol.~26,
  p.~702, jan 2018.

\bibitem{Macready1996}
W.~M. Macready, A.~G. Siapas, and S.~A. Kauffman, ``{Criticality and
  Parallelism in Combinatorial Optimization},'' {\em Science}, vol.~271,
  pp.~56--59, jan 1996.

\bibitem{LeCun2015}
Y.~LeCun, Y.~Bengio, and G.~Hinton, ``{Deep learning},'' {\em Nature},
  vol.~521, pp.~436--444, may 2015.

\bibitem{Little1974}
W.~A. Little, ``{The existence of persistent states in the brain},'' {\em
  Mathematical Biosciences}, vol.~19, no.~1-2, pp.~101--120, 1974.

\bibitem{Hopfield1982}
J.~J. Hopfield, ``{Neural networks and physical systems with emergent
  collective computational abilities.},'' {\em Proceedings of the National
  Academy of Sciences of the United States of America}, vol.~79, pp.~2554--8,
  apr 1982.

\bibitem{Hopfield}
J.~J. Hopfield and D.~W. Tank, ``{“Neural” computation of decisions in
  optimization problems},'' {\em Biological Cybernetics}, vol.~52, no.~3,
  pp.~141--152.

\bibitem{Ising1925}
E.~Ising, ``{Beitrag zur Theorie des Ferromagnetismus},'' {\em Z. Phys.}, 1925.

\bibitem{Mezard2009}
M.~M{\'{e}}zard and A.~Montanari, {\em {Information, Physics, and
  Computation}}.
\newblock 2009.

\bibitem{Amit1985}
D.~J. Amit, H.~Gutfreund, and H.~Sompolinsky, ``{Spin-glass models of neural
  networks},'' {\em Physical Review A}, vol.~32, pp.~1007--1018, aug 1985.

\bibitem{Pelissetto2002}
A.~Pelissetto and E.~Vicari, ``{Critical phenomena and renormalization-group
  theory},'' {\em Physics Reports}, vol.~368, pp.~549--727, oct 2002.

\bibitem{Onsager1944}
L.~Onsager, ``{Crystal Statistics. I. A Two-Dimensional Model with an
  Order-Disorder Transition},'' {\em Physical Review}, vol.~65, pp.~117--149,
  feb 1944.

\bibitem{Brilliantov1998}
N.~V. Brilliantov, ``{Effective magnetic Hamiltonian and Ginzburg criterion for
  fluids},'' {\em Physical Review E}, vol.~58, pp.~2628--2631, aug 1998.

\bibitem{Amit1989}
D.~J. Amit, {\em {Modeling brain function : the world of attractor neural
  networks}}.
\newblock Cambridge University Press, 1989.

\bibitem{Ghofraniha2015}
N.~Ghofraniha, I.~Viola, F.~{Di Maria}, G.~Barbarella, G.~Gigli, L.~Leuzzi, and
  C.~Conti, ``{Experimental evidence of replica symmetry breaking in random
  lasers},'' {\em Nature Communications}, vol.~6, p.~6058, dec 2015.

\bibitem{Halasz1998}
M.~A. Halasz, A.~D. Jackson, R.~E. Shrock, M.~A. Stephanov, and J.~J.~M.
  Verbaarschot, ``{Phase diagram of QCD},'' {\em Physical Review D}, vol.~58,
  p.~096007, sep 1998.

\bibitem{Barahona1982}
F.~Barahona, ``{On the computational complexity of Ising spin glass models},''
  {\em Journal of Physics A: Mathematical and General}, vol.~15,
  pp.~3241--3253, oct 1982.

\bibitem{Bruck1990}
J.~Bruck and J.~W. Goodman, ``{On the power of neural networks for solving hard
  problems},'' {\em Journal of Complexity}, vol.~6, pp.~129--135, jun 1990.

\bibitem{Farhat1985}
N.~H. Farhat, D.~Psaltis, A.~Prata, and E.~Paek, ``{Optical implementation of
  the Hopfield model},'' {\em Applied Optics}, vol.~24, p.~1469, may 1985.

\bibitem{Kirkpatrick1983}
S.~Kirkpatrick, C.~D. Gelatt, and M.~P. Vecchi, ``{Optimization by simulated
  annealing.},'' {\em Science (New York, N.Y.)}, vol.~220, pp.~671--80, may
  1983.

\bibitem{Earl2005}
D.~J. Earl and M.~W. Deem, ``{Parallel tempering: Theory, applications, and new
  perspectives},'' {\em Physical Chemistry Chemical Physics}, vol.~7, p.~3910,
  nov 2005.

\bibitem{Davis1991}
L.~D. Davis and M.~Mitchell, {\em {Handbook of Genetic Algorithms}}.
\newblock 1991.

\bibitem{Glover1998}
F.~Glover and M.~Laguna, ``{Tabu Search},'' in {\em Handbook of Combinatorial
  Optimization}, pp.~2093--2229, Boston, MA: Springer US, 1998.

\bibitem{Boros2007}
E.~Boros, P.~L. Hammer, and G.~Tavares, ``{Local search heuristics for
  Quadratic Unconstrained Binary Optimization (QUBO)},'' {\em Journal of
  Heuristics}, vol.~13, pp.~99--132, feb 2007.

\bibitem{Shen2017a}
Y.~Shen, N.~C. Harris, S.~Skirlo, M.~Prabhu, T.~Baehr-Jones, M.~Hochberg,
  X.~Sun, S.~Zhao, H.~Larochelle, D.~Englund, and M.~Solja{\v{c}}i{\'{c}},
  ``{Deep learning with coherent nanophotonic circuits},'' {\em Nature
  Photonics}, vol.~11, pp.~441--446, jun 2017.

\bibitem{Silva2014}
A.~Silva, F.~Monticone, G.~Castaldi, V.~Galdi, A.~Al{\`{u}}, and N.~Engheta,
  ``{Performing mathematical operations with metamaterials},'' {\em Science},
  vol.~343, no.~6167, pp.~160--163, 2014.

\bibitem{Koenderink2015}
A.~F. Koenderink, A.~Al{\`{u}}, and A.~Polman, ``{Nanophotonics: Shrinking
  light-based technology},'' may 2015.

\bibitem{Carolan2015}
J.~Carolan, C.~Harrold, C.~Sparrow, E.~Mart{\'{i}}n-L{\'{o}}pez, N.~J. Russell,
  J.~W. Silverstone, P.~J. Shadbolt, N.~Matsuda, M.~Oguma, M.~Itoh, G.~D.
  Marshall, M.~G. Thompson, J.~C. Matthews, T.~Hashimoto, J.~L. O'Brien, and
  A.~Laing, ``{Universal linear optics},'' {\em Science}, vol.~349,
  pp.~711--716, aug 2015.

\bibitem{Reck1994}
M.~Reck, A.~Zeilinger, H.~J. Bernstein, and P.~Bertani, ``{Experimental
  realization of any discrete unitary operator},'' {\em Physical Review
  Letters}, vol.~73, pp.~58--61, jul 1994.

\bibitem{Clements2016}
W.~R. Clements, P.~C. Humphreys, B.~J. Metcalf, W.~S. Kolthammer, and I.~A.
  Walsmley, ``{Optimal design for universal multiport interferometers},'' {\em
  Optica}, vol.~3, p.~1460, dec 2016.

\bibitem{Lin2018}
X.~Lin, Y.~Rivenson, N.~T. Yardimci, M.~Veli, M.~Jarrahi, and A.~Ozcan,
  ``{All-Optical Machine Learning Using Diffractive Deep Neural Networks},''
  apr 2018.

\bibitem{Gruber2000}
M.~Gruber, J.~Jahns, and S.~Sinzinger, ``{Planar-integrated optical
  vector-matrix multiplier},'' {\em Applied Optics}, vol.~39, p.~5367, oct
  2000.

\bibitem{Tait2014}
A.~N. Tait, M.~A. Nahmias, Y.~Tian, B.~J. Shastri, and P.~R. Prucnal,
  ``{Photonic Neuromorphic Signal Processing and Computing},'' pp.~183--222,
  Springer, Berlin, Heidelberg, 2014.

\bibitem{Tait}
A.~N. Tait, M.~A. Nahmias, B.~J. Shastri, and P.~R. Prucnal, ``{Broadcast and
  weight: An integrated network for scalable photonic spike processing},'' {\em
  Journal of Lightwave Technology}, vol.~32, no.~21, pp.~3427--3439, 2014.

\bibitem{Vandoorne2014}
K.~Vandoorne, P.~Mechet, T.~{Van Vaerenbergh}, M.~Fiers, G.~Morthier,
  D.~Verstraeten, B.~Schrauwen, J.~Dambre, and P.~Bienstman, ``{Experimental
  demonstration of reservoir computing on a silicon photonics chip},'' {\em
  Nature Communications}, vol.~5, p.~3541, dec 2014.

\bibitem{Saade2016}
A.~Saade, F.~Caltagirone, I.~Carron, L.~Daudet, A.~Dremeau, S.~Gigan, and
  F.~Krzakala, ``{Random projections through multiple optical scattering:
  Approximating Kernels at the speed of light},'' in {\em 2016 IEEE
  International Conference on Acoustics, Speech and Signal Processing
  (ICASSP)}, pp.~6215--6219, IEEE, mar 2016.

\bibitem{Pierangeli2018}
D.~Pierangeli, V.~Palmieri, G.~Marcucci, C.~Moriconi, G.~Perini, M.~{De
  Spirito}, M.~Papi, and C.~Conti, ``{Deep optical neural network by living
  tumour brain cells},'' dec 2018.

\bibitem{Cheng2014}
Z.~Cheng, H.~K. Tsang, X.~Wang, K.~Xu, and J.-B. Xu, ``{In-Plane Optical
  Absorption and Free Carrier Absorption in Graphene-on-Silicon Waveguides},''
  {\em IEEE Journal of Selected Topics in Quantum Electronics}, vol.~20,
  pp.~43--48, jan 2014.

\bibitem{Bao2011a}
Q.~Bao, H.~Zhang, Z.~Ni, Y.~Wang, L.~Polavarapu, Z.~Shen, Q.-H. Xu, D.~Tang,
  and K.~P. Loh, ``{Monolayer graphene as a saturable absorber in a mode-locked
  laser},'' {\em Nano Research}, vol.~4, pp.~297--307, mar 2011.

\bibitem{Selden1967}
A.~C. Selden, ``{Pulse transmission through a saturable absorber},'' {\em
  British Journal of Applied Physics}, vol.~18, pp.~743--748, jun 1967.

\bibitem{Soljacic2002}
M.~Solja{\v{c}}i{\'{c}}, M.~Ibanescu, S.~G. Johnson, Y.~Fink, and J.~D.
  Joannopoulos, ``{Optimal bistable switching in nonlinear photonic
  crystals},'' {\em Physical Review E}, vol.~66, p.~055601, nov 2002.

\bibitem{Schirmer1997}
R.~W. Schirmer and A.~L. Gaeta, ``{Nonlinear mirror based on two-photon
  absorption},'' {\em Journal of the Optical Society of America B}, vol.~14,
  p.~2865, nov 1997.

\bibitem{Horowitz1990}
P.~Horowitz and H.~Winfield, ``{The Art of Electronics},'' {\em American
  Journal of Physics}, 1990.

\bibitem{Boser1991}
B.~Boser, E.~Sackinger, J.~Bromley, Y.~{Le Cun}, and L.~Jackel, ``{An analog
  neural network processor with programmable topology},'' {\em IEEE Journal of
  Solid-State Circuits}, vol.~26, no.~12, pp.~2017--2025, 1991.

\bibitem{Misra2010}
J.~Misra and I.~Saha, ``{Artificial neural networks in hardware: A survey of
  two decades of progress},'' {\em Neurocomputing}, vol.~74, pp.~239--255, dec
  2010.

\bibitem{Vrtaric2013}
D.~Vrtaric, V.~Ceperic, and A.~Baric, ``{Area-efficient differential Gaussian
  circuit for dedicated hardware implementations of Gaussian function based
  machine learning algorithms},'' {\em Neurocomputing}, vol.~118, pp.~329--333,
  oct 2013.

\bibitem{Williamson2019}
I.~A.~D. Williamson, T.~W. Hughes, M.~Minkov, B.~Bartlett, S.~Pai, and S.~Fan,
  ``{Reprogrammable Electro-Optic Nonlinear Activation Functions for Optical
  Neural Networks},'' {\em arXiv preprints, arxiv.org:1903.04579}, mar 2019.

\bibitem{Rendl2010}
F.~Rendl, G.~Rinaldi, and A.~Wiegele, ``{Solving Max-Cut to optimality by
  intersecting semidefinite and polyhedral relaxations},'' {\em Mathematical
  Programming}, vol.~121, pp.~307--335, feb 2010.

\bibitem{Peretto1984}
P.~Peretto, ``{Collective properties of neural networks: A statistical physics
  approach},'' {\em Biological Cybernetics}, vol.~50, pp.~51--62, feb 1984.

\bibitem{Lipson2005}
M.~Lipson, ``{Guiding, Modulating, and Emitting Light on Silicon-Challenges and
  Opportunities},'' {\em Journal of Lightwave Technology, Vol. 23, Issue 12,
  pp. 4222-}, vol.~23, p.~4222, dec 2005.

\bibitem{Harris2014}
N.~C. Harris, Y.~Ma, J.~Mower, T.~Baehr-Jones, D.~Englund, M.~Hochberg, and
  C.~Galland, ``{Efficient, compact and low loss thermo-optic phase shifter in
  silicon},'' {\em Optics Express}, vol.~22, p.~10487, may 2014.

\bibitem{Hamerly2018}
R.~Hamerly, T.~Inagaki, P.~L. McMahon, D.~Venturelli, A.~Marandi, T.~Onodera,
  E.~Ng, C.~Langrock, K.~Inaba, T.~Honjo, K.~Enbutsu, T.~Umeki, R.~Kasahara,
  S.~Utsunomiya, S.~Kako, K.~I. Kawarabayashi, R.~L. Byer, M.~M. Fejer,
  H.~Mabuchi, D.~Englund, E.~Rieffel, H.~Takesue, and Y.~Yamamoto,
  ``{Experimental investigation of performance differences between coherent
  Ising machines and a quantum annealer},'' {\em Science Advances}, 2019.

\bibitem{Miller2015a}
D.~A.~B. Miller, ``{Perfect optics with imperfect components},'' {\em Optica},
  vol.~2, p.~747, aug 2015.

\bibitem{Burgwal2017}
R.~Burgwal, W.~R. Clements, D.~H. Smith, J.~C. Gates, W.~S. Kolthammer, J.~J.
  Renema, and I.~A. Walmsley, ``{Using an imperfect photonic network to
  implement random unitaries},'' {\em Optics Express}, vol.~25, p.~28236, nov
  2017.

\bibitem{Almeida2004}
V.~R. Almeida, C.~A. Barrios, R.~R. Panepucci, and M.~Lipson, ``{All-optical
  control of light on a silicon chip},'' {\em Nature}, vol.~431,
  pp.~1081--1084, oct 2004.

\bibitem{Phare2015}
C.~T. Phare, Y.~H. {Daniel Lee}, J.~Cardenas, and M.~Lipson, ``{Graphene
  electro-optic modulator with 30 GHz bandwidth},'' {\em Nature Photonics},
  2015.

\bibitem{Haffner2018}
C.~Haffner, D.~Chelladurai, Y.~Fedoryshyn, A.~Josten, B.~Baeuerle, W.~Heni,
  T.~Watanabe, T.~Cui, B.~Cheng, S.~Saha, D.~L. Elder, L.~R. Dalton,
  A.~Boltasseva, V.~M. Shalaev, N.~Kinsey, and J.~Leuthold, ``{Low-loss
  plasmon-assisted electro-optic modulator},'' {\em Nature}, 2018.

\bibitem{Hamerly2018b}
R.~Hamerly, L.~Bernstein, A.~Sludds, M.~Solja{\v{c}}i{\'{c}}, and D.~Englund,
  ``{Large-Scale Optical Neural Networks Based on Photoelectric
  Multiplication},'' {\em Physical Review X}, vol.~9, nov 2019.

\bibitem{Metropolis1953}
N.~Metropolis, A.~W. Rosenbluth, M.~N. Rosenbluth, A.~H. Teller, and E.~Teller,
  ``{Equation of state calculations by fast computing machines},'' {\em The
  Journal of Chemical Physics}, 1953.

\bibitem{Hastings1970}
W.~K. Hastings, ``{Monte carlo sampling methods using Markov chains and their
  applications},'' {\em Biometrika}, 1970.

\bibitem{Dean2018}
J.~Dean, D.~Patterson, and C.~Young, ``{A New Golden Age in Computer
  Architecture: Empowering the Machine-Learning Revolution},'' {\em IEEE
  Micro}, vol.~38, pp.~21--29, mar 2018.

\bibitem{Dou2005}
Y.~Dou, S.~Vassiliadis, G.~K. Kuzmanov, and G.~N. Gaydadjiev, ``{64-bit
  floating-point FPGA matrix multiplication},'' in {\em Proceedings of the 2005
  ACM/SIGDA 13th international symposium on Field-programmable gate arrays -
  FPGA '05}, (New York, New York, USA), p.~86, ACM Press, 2005.

\bibitem{Knill2004}
D.~C. Knill and A.~Pouget, ``{The Bayesian brain: the role of uncertainty in
  neural coding and computation},'' {\em Trends in Neurosciences}, vol.~27,
  pp.~712--719, dec 2004.

\bibitem{Maass2014}
W.~Maass, ``{Noise as a resource for computation and learning in networks of
  spiking neurons},'' {\em Proceedings of the IEEE}, vol.~102, pp.~860--880,
  may 2014.

\bibitem{Wang2016}
Q.~Wang, E.~T.~F. Rogers, B.~Gholipour, C.-M. Wang, G.~Yuan, J.~Teng, and N.~I.
  Zheludev, ``{Optically reconfigurable metasurfaces and photonic devices based
  on phase change materials},'' {\em Nature Photonics}, vol.~10, pp.~60--65,
  jan 2016.

\bibitem{Tait2017}
A.~N. Tait, T.~F. de~Lima, E.~Zhou, A.~X. Wu, M.~A. Nahmias, B.~J. Shastri, and
  P.~R. Prucnal, ``{Neuromorphic photonic networks using silicon photonic
  weight banks},'' {\em Scientific Reports}, vol.~7, p.~7430, dec 2017.

\bibitem{Karp1972}
R.~M. Karp, ``{Reducibility among Combinatorial Problems},'' in {\em Complexity
  of Computer Computations}, pp.~85--103, Boston, MA: Springer US, 1972.

\bibitem{Looi1992}
C.-K. Looi, ``{Neural network methods in combinatorial optimization},'' {\em
  Computers {\&} Operations Research}, vol.~19, pp.~191--208, apr 1992.

\bibitem{Azzalini1996}
A.~Azzalini and A.~D. VALLE, ``{The multivariate skew-normal distribution},''
  {\em Biometrika}, vol.~83, pp.~715--726, dec 1996.

\bibitem{Friedli2017}
S.~Friedli and Y.~Velenik, {\em {Statistical Mechanics of Lattice Systems}}.
\newblock Cambridge University Press, nov 2017.

\bibitem{Harris2017}
N.~C. Harris, G.~R. Steinbrecher, M.~Prabhu, Y.~Lahini, J.~Mower, D.~Bunandar,
  C.~Chen, F.~N.~C. Wong, T.~Baehr-Jones, M.~Hochberg, S.~Lloyd, and
  D.~Englund, ``{Quantum transport simulations in a programmable nanophotonic
  processor},'' {\em Nature Photonics}, vol.~11, pp.~447--452, jun 2017.

\bibitem{Pierangeli2019}
D.~Pierangeli, G.~Marcucci, and C.~Conti, ``{Large-Scale Photonic Ising Machine
  by Spatial Light Modulation},'' {\em Physical Review Letters}, 2019.

\bibitem{Nourani1998}
Y.~Nourani and B.~Andresen, ``{A comparison of simulated annealing cooling
  strategies},'' {\em Journal of Physics A: Mathematical and General}, vol.~31,
  pp.~8373--8385, oct 1998.

\bibitem{Cohn1999}
H.~Cohn and M.~Fielding, ``{Simulated Annealing: Searching for an Optimal
  Temperature Schedule},'' {\em SIAM Journal on Optimization}, vol.~9,
  pp.~779--802, jan 1999.

\end{thebibliography}
\end{document}